\documentclass[a4paper,12pt]{article}
\usepackage{amssymb}
\usepackage{amsmath}
\usepackage{mathtools}
\usepackage{bm}
\usepackage{hyperref}
\usepackage{comment}
\usepackage{enumerate}
\usepackage{here}
\usepackage{latexsym}
\usepackage{array}
\usepackage{mathrsfs}
\usepackage{geometry}
\usepackage{fancyhdr}
\renewcommand{\title}[1]{

\begin{center} \Large \bf #1 \end{center}
}

\renewcommand{\author}[2]{
 \begin{center} #1  \vspace{3mm} \\
  #2 \\
 \end{center}
\addvspace{\baselineskip}
}

\usepackage{amsthm}
\newtheorem{theorem}{Theorem}[section]
\newtheorem{proposition}[theorem]{Proposition}

\theoremstyle{definition}

\theoremstyle{remark}

\makeatletter
\@addtoreset{equation}{section}

\makeatother

\usepackage{graphicx}

\def\propagator{\includegraphics[width=0.1\textwidth]{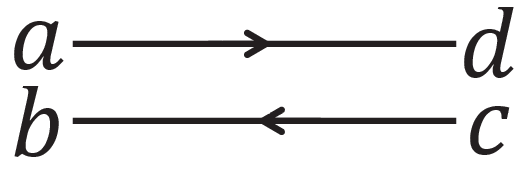}}

\def\vertexs{\includegraphics[width=0.08\textwidth]{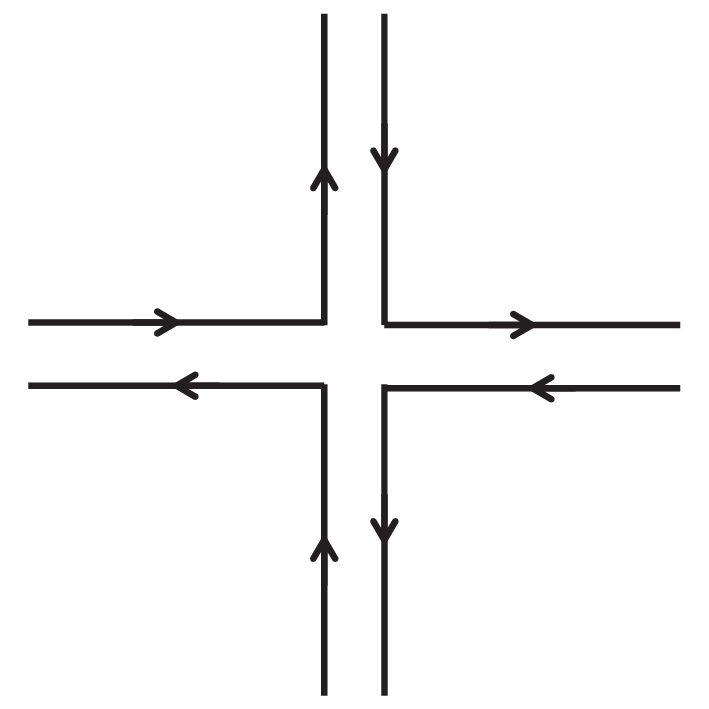}}

\def\8{\includegraphics[width=1.0\textwidth]{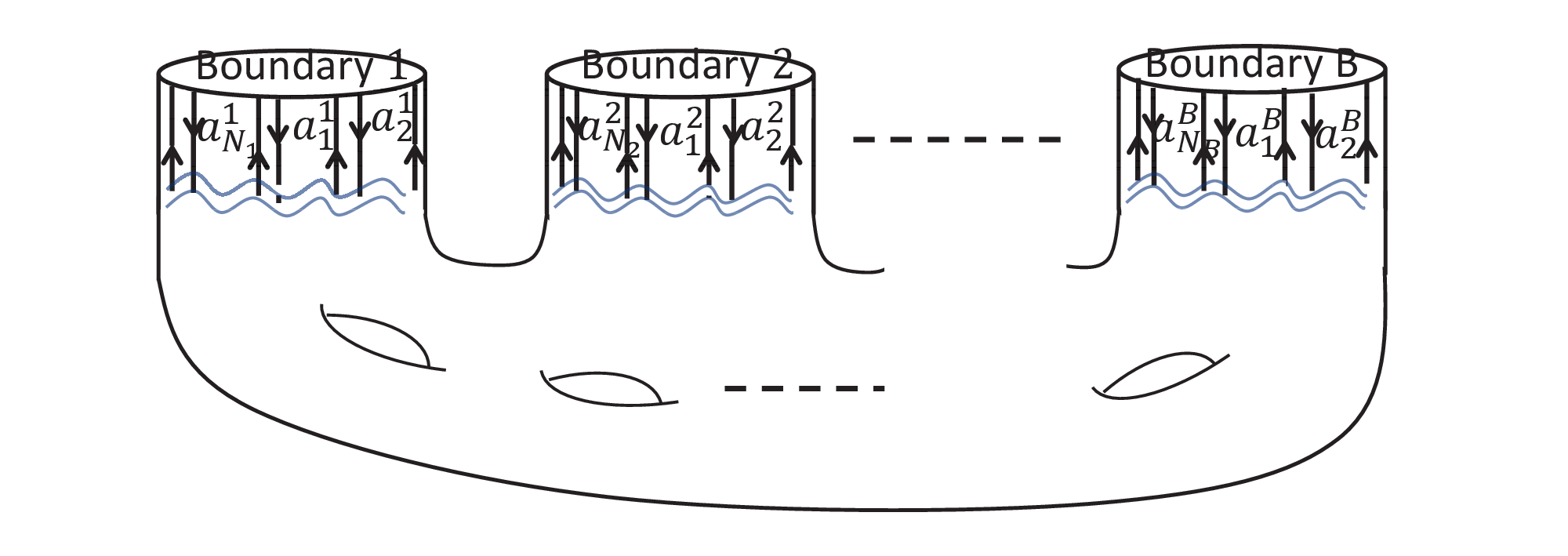}}

\begin{document}
\baselineskip 5mm
\title{Relationship between a $\Phi^4$ matrix model and harmonic oscillator systems}
\author{${}^1$Harald Grosse,  ${}^2$Naoyuki Kanomata, ${}^2$Akifumi Sako, ${}^3$Raimar Wulkenhaar
}
{${}^{1}$
Erwin Schr\"odinger International Institute for Mathematics and Physics, \\
University of Vienna, Boltzmanngasse 9, 1090 Vienna, Austria \vspace{3mm}\\

${}^1$
Faculty of Physics, University of Vienna, Boltzmanngasse 5, 
1090 Vienna, Austria
\vspace{3mm}\\

${}^2$
Tokyo University of Science, 1-3 Kagurazaka, Shinjuku-ku, Tokyo, 162-8601, Japan
\vspace{3mm}\\

${}^3$
Mathematisches Institut, Universit\"at M\"unster,
Einsteinstra{\ss}e 62, D-48149 M\"unster, Germany
}
\noindent
\vspace{1cm}

\abstract{A Hermitian $\Phi^4$ matrix model with a Kontsevich-type
  kinetic term is studied.  It was recently discovered that the
  partition function of this matrix model satisfies the
  Schr\"odinger equation of the $N$-body harmonic oscillator, and that
  eigenstates of the Virasoro operators can be derived from this
  partition function.  We extend these results and obtain an explicit
  formula for such eigenstates in terms of the free
  energy. Furthermore, the Schr\"odinger equation for the $N$-body
  harmonic oscillator can also be reformulated in terms of connected
  correlation functions. The $U(1)^N$-symmetry allows us to
 derive loop equations.  }
%
%
%

\allowdisplaybreaks
\section{Introduction}\label{sect1}

In the 1990s, numerous connections between matrix models and
two-dimensional quantum gravity were discovered, and many important
developments have been made. We refer to \cite{DiFrancesco:1993cyw}
for an early review that covers most of these achievements.
Of particular importance is the Kontsevich model
\cite{Kontsevich:1992ti}. 
Its action is given by
$ S_{K} = N~ Tr \{ E \Phi^2 + \frac{\lambda}{3} \Phi^3 \} , $ where
$\Phi$ is an $N\times N$ Hermitian matrix, $E$ is a positive diagonal
$N\times N$ matrix $E := diag (E_1, E_2 , \cdots ,E_N )$ without
degenerate eigenvalues, and $\lambda$ is a complex number as a
coupling constant.  This model was proposed to prove the Witten
conjecture \cite{Witten:1990hr}.  The model we will study in this paper
is given by
replacing the interaction term $\frac{1}{3} \Phi^3$ by
$\frac{1}{4} \Phi^4$.
The motivation to consider such kind of matrix models 
comes from quantum field theories on noncommutative spaces.

A noncommutative space is studied in many ways.  For example, we
consider a noncommutative function algebra as a noncommutative space
obtained by deforming a commutative function algebra into a
noncommutative one.  In the process, we replace the commutative
product with a noncommutative product.  Then, by using an appropriate
matrix representation to the noncommutative algebra, the field theory
can be described as a matrix model.

Most of quantum field theories on noncommutative spaces are not
renormalizable because of the UV/IR problem.  However, the scalar
$\Phi^4$ field theory on the Moyal space proposed by Grosse and Wulkenhaar 
\cite{Grosse:2004yu} and the scalar $\Phi^3$ theory deceloped 
by Grosse and
Steinacker \cite{Grosse:2005ig,Grosse:2006tc} appeared as 
exceptions that could be renormalized.  In other words, they showed
that those field theories to which certain counter-Lagrangian terms
are added are renormalizable.  This scalar $\Phi^3$ theory is
basically equivalent to the Kontsevich model.  This scalar $\Phi^4$
quantum field theory on the Moyal space (Grosse-Wulkenhaar model) is the
model we discuss in this paper.  (It is also worth noting here that
alternative formulations of renormalizable quantum field theory on noncommutative
spaces, different from the Grosse-Wulkenhaar or Grosse-Steinacker
types, have also been discussed in recent years \cite{Nguyen:2021rsa}.)

The $\Phi^3$ theory (the Kontsevich model) has been known from the
outset to correspond to the KdV hierarchy, and it has been expected
that a model in which the interactions are simply replaced by $\Phi^4$
would also be related to integrable systems.  This is because, if the
Feynman diagrams of the $\Phi^3$ matrix model are mapped onto
triangulations of a surface, it is natural to think that the $\Phi^4$
matrix model would simply correspond to quadrangulation, with no
essential difference.  Actually, it was found that the partition
functions of the Hermitian $\Phi^4$-matrix model correspond to
zero-energy solutions of a Schr\"odinger-type equation with 
an $N$-body
harmonic oscillator Hamiltonian.  Furthermore, the partition functions
of the real symmetric $\Phi^4$-matrix model corresponds to zero-energy
solutions of a Schr\"odinger type equation with the Calogero-Moser
Hamiltonian \cite{Grosse:2023jcb,Grosse:2023ncx,Grosse:2024bfh}.

\medskip

In this paper, we extend the result that the partition function
satisfies the Schr\"odinger equation for the $N$-body harmonic
oscillator system, and obtain an explicit formula for such eigenstates
in terms of the free energy. Furthermore, since the free energy serves
as the generating function for connected multi-point correlation
functions, the differential equation for the harmonic oscillator can
also be reformulated in terms of these connected correlation
functions.  The corresponding equations for the connected two- and
four-point functions are derived.  These results are further confirmed
perturbatively up to first order in the coupling constant of the
interaction.  The process of obtaining the Schr\"odinger equation for
the $N$-body harmonic oscillator is constructed from a set of
Schwinger-Dyson equations.  The contribution of additional
Schwinger-Dyson equations is often discussed using loop equations in
Hermitian matrix models with $U(N)$ symmetry.  Although this model
lacks $U(N)$ symmetry due to the presence of a kinetic term, it
retains $U(1)^N$ symmetry, enabling us to derive equations similar to
loop equations, as described in \cite{Eynard:2015aea}.

\section{Setup of $\Phi^{4}$ Matrix Model and preparations}\label{sec2}

In this section, we review the $\Phi^{4}$ matrix model based on previous studies \cite{Grosse:2023jcb,Grosse:2023ncx,Grosse:2024bfh}, and we provide the notation in this paper.

Let $\Phi=(\Phi_{ij})$ be a Hermitian matrix for $i,j=1,2,\ldots,N$ and $E$ be a real
 diagonal $N\times N$ matrix 
$E := diag (E_1, E_2 , \cdots ,E_N )$ without
degenerate eigenvalues, i.e. $E_i \neq E_j $ if $i \neq j$.
Let us consider the following action:
\begin{align}
S[\Phi]=N\mathrm{tr}\left(E\Phi^{2}+\frac{\eta}{4}\Phi^{4}\right),
\end{align}
where $\eta$ is a coupling constant that is a positive real number. Since the diagonal matrix $E$ is not proportional to the unit matrix in general, there is no symmetry for the unitary transformation in $\Phi\rightarrow U\Phi U^{\dagger}$. Here $U$ is a unitary matrix, and $U^{\dagger}$ is its Hermitian conjugate. 

Let $\mathcal{D}\Phi$ be the ordinary Haar measure,
\begin{align}
\displaystyle\mathcal{D}\Phi:=&\prod_{i=1}^{N}d\Phi_{ii}\prod_{1\leq i<j\leq N}d\mathrm{Re}\Phi_{ij}d\mathrm{Im}\Phi_{ij},
\end{align}   
where each variable is divided into real and imaginary parts $\Phi_{ij}=\mathrm{Re}\Phi_{ij}+i\mathrm{Im}\Phi_{ij}$ with $\mathrm{Re}\Phi_{ij}=\mathrm{Re}\Phi_{ji}$ and $\mathrm{Im}\Phi_{ij}=-\mathrm{Im}\Phi_{ji}$. Let us consider the following partition function:
\begin{align}
{Z}(E, \eta):=&\int_{H_N} \mathcal{D}\Phi\exp\left(-S[\Phi]\right), \label{partitionfunction}
\end{align}
where $H_N$ is the space of $N\times N$ Hermitian matrices.

Let $\Delta(E)$ be the Vandermonde  determinant 
$\Delta (E) := \prod_{k<l} (E_l -E_k)$.
Then the function
\begin{align*}
\Psi (E, \eta ) := 
e^{-\frac{N}{2 \eta}\sum_{i=1}^{N}E_{i}^{2}}\Delta(E)
{Z}(E,\eta)
\end{align*}
is a zero-energy solution of the Schr\"odinger type equation
for the 
$N$-body harmonic oscillator system.

\begin{theorem}\label{main_thm}\cite{Grosse:2023jcb}
Let $\Psi (E, \eta )$ be the function defined above.
Then $\Psi (E, \eta )$ is a zero-energy solution of 
the Schr\"odinger type equation
\begin{align*}
{\mathcal H}_{HO} \Psi (E, \eta ) = 0. 
\end{align*}
Here ${\mathcal H}_{HO}$ is the Hamiltonian for the 
$N$-body harmonic oscillator system
\begin{align}\label{N_harmonic_H}
{\mathcal H}_{HO}:= - \frac{\eta}{N} \sum_{i=1}^N 
\left( 
\frac{\partial}{\partial E_i}
\right)^2  + \frac{N}{\eta}\sum_{i=1}^N (E_i)^2 .
\end{align}
\end{theorem}
It is also known that by replacing the Hermitian matrix $\Phi$ with a real symmetric matrix, 
the Hamiltonian of the above harmonic oscillator is replaced with the Calogero model Hamiltonian \cite{Grosse:2023ncx}.

\medskip

With the aim of naturally introducing the formulae we will use later,
we give in this section a rough outline of the proof of this theorem.
Details are given in \cite{Grosse:2023jcb}. 
To derive the above differential equation,
we introduce $H$ as a positive  Hermitian 
$N\times N$ matrix
with nondegenerate eigenvalues
$\{E_1, E_2 , \cdots ,E_N ~ | ~ E_i \neq E_j ~\mbox{for}~ i \neq j \}$.
Using this $H$, we consider the new action
\begin{align}
S&= N~ Tr \{ H \Phi^2 + \frac{\eta}{4} \Phi^4  \}
\notag \\
&= N \left( 
\sum_{i,j,k}^N  H_{ij}\Phi_{jk}\Phi_{ki}
+ \frac{\eta}{4} \sum_{i,j,k,l}^N
\Phi_{ij}\Phi_{jk}\Phi_{kl}\Phi_{li}
\right).
\label{action_S}
\end{align}
The partition function defined by this $S$ 
\begin{align}
Z(E, \eta) := \int_{H_N} \mathcal{D} \Phi ~e^{-S} , \label{partitionfunction_2}
\end{align}
is the same one defined by (\ref{partitionfunction}),
because the integral measure is $U(N)$ invariant.
We use the symbol $\displaystyle \langle O \rangle$
as a non-normalized vacuum expectation value
defined by
$\displaystyle \langle O \rangle := \int_{H_N} \mathcal{D} \Phi ~ O e^{-S} $.

The Schwinger-Dyson equation is derived from 
\begin{align}
\int_{H_N} \mathcal{D} \Phi ~  \frac{\partial}{\partial \Phi_{ij}}
\left(
\Phi_{ij} e^{-S}
\right) = 0,
\end{align}
which is expressed as
\begin{align}
Z(E, \eta) - N \sum_k (\langle
H_{ki}\Phi_{ij}\Phi_{jk}
\rangle 
+\langle
H_{jk}\Phi_{ki}\Phi_{ij}
\rangle 
)
-N\eta 
 \sum_{k,l} \langle
\Phi_{jk}\Phi_{kl}\Phi_{li}\Phi_{ij}
\rangle =0 . \label{eq8}
\end{align}
To obtain the desired partial differential equation, we use the fact that the following expectation values can be expressed in terms of partial derivatives:
\begin{align}
\frac{\partial Z(E, \eta) }{\partial H_{ij}} =
-N \sum_k \langle  \Phi_{jk}\Phi_{ki} \rangle ,
\quad
 \frac{\partial^2 Z(E, \eta) }{\partial H_{ij}\partial H_{mn}} =
N^2  \sum_{k,l} \langle \Phi_{jk}\Phi_{ki}
\Phi_{nl}\Phi_{lm} \rangle .  \label{eq9}
\end{align}
After summing (\ref{eq8}) over
indices $i$ and $j$ 
and substituting (\ref{eq9}) for it,
\begin{align}
{\mathcal L}_{SD}^H Z(E, \eta) = 0 . \label{SD_H}
\end{align}
is obtained, where
${\mathcal L}_{SD}^H $ is a second order differential operator
defined by
\begin{align}
{\mathcal L}_{SD}^H:=
N^2 + 2 \sum_{i,k} H_{ki} \frac{\partial  }{\partial H_{ki}}
-\frac{\eta}{N} \sum_{i,k}
\left( 
\frac{\partial  }{\partial H_{ki}}\frac{\partial  }{\partial H_{ik}}
\right) .
\end{align}

To rewrite this Schwinger-Dyson equation in terms of  
$E_n (n= 1,2, \cdots , N)$, the following formulae for
the second term,
\begin{align}
\sum_{i,j} H_{ij} \frac{\partial  Z(E, \eta)}{\partial H_{ij}}
= \sum_k E_k  \frac{\partial  Z(E, \eta)}{\partial E_{k}} ,
\label{H-E_formula2}
\end{align}
and the third term,
\begin{align}
\sum_{i,k}
\left( 
\frac{\partial  }{\partial H_{ki}}\frac{\partial  }{\partial H_{ik}}
\right)~Z(E, \eta)
&=
\left\{
\sum_{i=1}^N \left( \frac{\partial}{\partial E_i} \right)^2
+\sum_{i \neq j} \frac{1}{E_i - E_j}
\left( \frac{\partial}{\partial E_i} - 
\frac{\partial}{\partial E_j} \right)
\right\}~ Z(E, \eta)   \label{H-E_laplace}
\end{align}
are used.
Here $\displaystyle \sum_{i \neq j}$ means 
$\displaystyle \sum_{i,j=1 , i \neq j}^N$.
From (\ref{SD_H}) ,  (\ref{H-E_formula2}) , and (\ref{H-E_laplace}),
we find that
the partition function defined by (\ref{partitionfunction})
satisfies 
\begin{align}
{\mathcal L}_{SD} Z(E, \eta) = 0 ,
\label{SD_2}
\end{align}
where 
\begin{align}
{\mathcal L}_{SD} :=
\left\{
\frac{\eta}{N} \sum_{i=1}^N \left( \frac{\partial}{\partial E_i} \right)^2
+
\frac{\eta}{N} \sum_{i \neq j} \frac{1}{E_i - E_j}
\left( \frac{\partial}{\partial E_i} - 
\frac{\partial}{\partial E_j} \right)
-2 \sum_k E_k  \frac{\partial }{\partial E_{k}} -N^2
\right\}~. \label{LSD}
\end{align}

We can diagonalize ${\mathcal L}_{SD}$
by some kind of gauge transformation as
\begin{align}\label{prop3_1}
-e^{-\frac{N}{2\eta} \sum_i E_i^2} \Delta(E)
{\mathcal L}_{SD} 
\Delta^{-1}(E) e^{\frac{N}{2\eta} \sum_i E_i^2} 
= -\frac{\eta}{N} \sum_{i=1}^N 
\left( 
\frac{\partial}{\partial E_i}
\right)^2  + \frac{N}{\eta}\sum_{i=1}^N (E_i)^2 .
\end{align}
The right-hand side is the Hamiltonian of the $N$-body harmonic oscillator
system and we denote it 
by ${\mathcal H}_{HO}$
as (\ref{N_harmonic_H}).
%
To cancel the gauge transformation $\Delta^{-1}(E) e^{\frac{N}{2\eta} \sum_i E_i^2} =:g$,
we introduce a transformed partition function $\Psi (E, \eta )$ by
\begin{align}\label{Psi_def_3_8}
\Psi (E, \eta ) := 
e^{-\frac{N}{2\eta} \sum_i E_i^2} \Delta(E) Z(E, \eta ) = g^{-1}  Z(E, \eta ) .
\end{align}
From (\ref{prop3_1}), we find that
$\Psi (E, \eta ) $
is a zero-energy solution of the 
Schr\"odinger-type differential equation:
\begin{align}
{\mathcal H}_{HO} \Psi(E, \eta ) = 0.
\label{Schrodinger}
\end{align}
Theorem \ref{main_thm} 
is thus proved.
This $N$-body harmonic oscillator system has
no interaction terms between the oscillators, so
it is a trivial quantum integrable system.

\bigskip

We have one remark here.
It is known that there are no zero eigenvalue solutions of the $N$-body harmonic oscillator system
in $L^2 (\mathbb{R} ^N) $, so the solution obtained here is not such a function \cite{Grosse:2024bfh}.

\bigskip

We will also use the partition function and the free energy with external fields, so let us introduce them here.
Let $J=(J_{mn})$ be a Hermitian matrix for $m,n=1,\ldots,N$ as an external field. 
Let us consider the following partition function and the free energy
with this $J$:
\begin{align}
\mathcal{Z}[E, J]:=&\int\mathcal{D}\Phi\exp\left(-S[\Phi]+N\mathrm{tr}(J\Phi)\right)\notag\\
=&\int \mathcal{D}\Phi \exp\left(-N\mathrm{tr}\left(E\Phi^2+\frac{\eta}{4}\Phi^{4}\right)\right)\exp\left(N\mathrm{tr}\left(J\Phi\right)\right).\label{Z[J]} \\
F[E , J] :=& \log \mathcal{Z}[E, J] 
\end{align}
Note that $\mathcal{Z}[E, 0] = Z(E, \eta)$.\\

\bigskip

We saw that the partition function corresponds to 
the solution of the Schr\"odinger equation for the $N$-body harmonic oscillator.
It is known that the harmonic oscillator system leads to a representation of the Virasoro (Witt) algebra.
We see that this yields an infinite sequence of partial differential equations to be satisfied by the partition function.\\

For simplicity, we use variables $y_i := \sqrt{\frac{N}{\eta}}E_i$. The Hamiltonian (\ref{N_harmonic_H}) in this coordinate is written as
\begin{align*}
\mathcal{H}_{{HO}} &= 
\sum_{i=1}^{N} \left( -\frac{\partial^2}{\partial y_i^2} 
+ y_i^2 \right) =
 \sum_{i=1}^{N} \{ a_i ~, ~ a_i^\dagger \}, 
\end{align*}
where \begin{align*} a_i &= \frac{1}{\sqrt{2}} \left( y_i + \frac{\partial}{\partial y_i} \right), \quad a_i^\dagger = \frac{1}{\sqrt{2}} \left( y_i - \frac{\partial}{\partial y_i} \right) .
\end{align*}

We introduce the Virasoro generators with a free parameter $\alpha$:
\begin{align}
L_{-n} &= \sum_{i=1}^{N} 
\left( \alpha \left( a_i^\dagger \right)^{n+1} a_i 
+ (1-\alpha) a_i \left( a_i^\dagger \right)^{n+1} \right) ,
\ \ 
(n \ge -1 ) \label{L_n}
\end{align}
which satisfy the following commutation relations:
\begin{align*}
\left[ L_n, L_m \right] &= (n-m) L_{n+m} .
\end{align*}
In particular, $L_0$ is written by using $\mathcal{H}_{{HO}} $:
\begin{align*}
L_0 
&= \frac{1}{2} \mathcal{H}_{{HO}} + \left(\frac{1}{2}-\alpha \right) N  .
\end{align*}
Then, we find that $ \mathcal{H}_{{HO}} $ satisfies
\begin{align*}
\left[ \frac{1}{2} \mathcal{H}_{{HO}} , L_{-m} \right] &= m L_{-m} . 
\end{align*}
Using $\displaystyle g = 
\Delta^{-1}(E) e^{\frac{N}{2\eta} \sum_i E_i^2}  
=\left( \frac{N}{\eta} \right)^\frac{N(N-1)}{4} \!\!\!\! \Delta^{-1}(y) e^{\frac{1}{2} \sum_i y_i^2}  $, we define $\tilde{L}_{n}$ by
$\tilde{L}_{n}:= g L_n g^{-1}$ {satisfying}
$ 
\left[ \tilde{L}_n, \tilde{L}_m \right] = (n-m) \tilde{L}_{n+m} $.
%
%
Using this  $\tilde{L}_{n}$, we immediately obtain the following result,
\begin{align*}
\left[ \mathcal{L}_{{SD}}, \tilde{L}_{-m} \right] &= -2 \left[ \tilde{L}_0, \tilde{L}_{-m} \right] = -2m \tilde{L}_{-m} ,
\end{align*}
 which implies the following theorem.
\begin{theorem}\label{main2}
The partition function defined by (\ref{partitionfunction}) satisfies
\begin{align}
\mathcal{L}_{SD}(\widetilde{L}_{-m}{Z}(E,\eta))=&-2m(\widetilde{L}_{-m}{Z}(E,\eta)) \quad (m \ge -1 ).
\end{align}
\end{theorem}

\section{Schwinger-Dyson Equation for free energy}\label{sect3}
In this section, we will rewrite the results of Section \ref{sec2} using the free energy, 
which is the generating function of the connected Green's function.

As we saw in (\ref{SD_2})
the partition function $\mathcal{Z}[E,0] = Z(E, \eta)$
satisfies the partial differential equation
$
{\mathcal L}_{SD}\mathcal{Z}[E, 0]= 0 .
$

Let $F[E, J]:=\log\mathcal{Z}[E, J]$ be the free energy. Then the following is obtained immediately.
\begin{proposition}
The Schwinger-Dyson equation
${\mathcal L}_{SD}\mathcal{Z}[E, 0]= 0$ is equivalent to
\begin{align}
&\frac{\eta}{N}\sum_{i=1}^{N}\left(\frac{\partial^{2}}{\partial E_{i}^{2}}F[E, 0]\right)
+\frac{\eta}{N}\sum_{i=1}^{N}\left(\frac{\partial}{\partial E_{i}}F[E, 0]\right)\left(\frac{\partial}{\partial E_{i}}F[E, 0]\right)
\notag\\
&+\frac{\eta}{N}\sum_{i,j=1,i\neq j}^{N}\frac{1}{E_{i}-E_{j}}\left(\frac{\partial}{\partial E_{i}}F[E, 0]
-\frac{\partial}{\partial E_{j}}F[E, 0]\right)
-2\sum_{k=1}^{N}E_{k}\left(\frac{\partial}{\partial E_{k}}F[E, 0]\right)-N^{2}=0 .
\label{SDeq_FreeEnergy}
\end{align}
\end{proposition}

In the following, the eigenfunctions of $\mathcal{H}_{HO}$ with eigenvalue $-2m$ given 
by Theorem \ref{main2}, will also be rewritten using the free energy.

\begin{proposition}


Let us introduce $\Psi_m := L_{-m} (Z[E,0] g^{-1}) =  L_{-m} \Psi$
$(m= 0,1,2,\cdots )$. 
\begin{align}
\mathcal{H}_{HO} \Psi_m  =2m \Psi_m \label{t}
\end{align}

\end{proposition}

\begin{proof}		
From the Schwinger-Dyson equation $\hat{\mathcal{L}}_{SD} Z[E,0]=0 ,$ or 
equivalent equation $\mathcal{H}_{HO} \Psi =0$, and $[ \mathcal{H}_{HO} , L_{-m} ] =  2m L_{-m}$,
this proposition follows immediately.
\end{proof}

\begin{proposition}\label{Prop3_3}
For a non-negative integer $m$, $\Psi_m = L_{-m}e^{F[E,0]}g^{-1}$ is given by
%
%
\begin{align}
\Psi_m 
=&-\left(-\frac{1}{\sqrt{2}}\right)^{m}\left(\frac{\eta}{N}\right)^{\frac{N(N-1)}{4}} e^{\frac{1}{2}\sum_{k=1}^{N}y_{k}^{2}} \times
\notag \\
& \sum_{l=1}^{N} \Biggl\{ y_{l}\left(\frac{\partial}{\partial y_{l}}\right)^{m+1} 
 +\frac{1}{{2}}  \left(\frac{\partial}{\partial y_{l}}\right)^{m+2} 
+ \alpha(m+1) \left(\frac{\partial}{\partial y_{l}}\right)^{m} \Biggl\}
e^{F[E,0]-V[E]}.
\end{align}
Here $V[y]:= V[E]:= \sum_{k=1}^{N}y_{k}^{2}+\sum_{1\leq i<j\leq N}\log\left(y_{j}-y_{i}\right)$ with $\displaystyle y_i = \sqrt{\frac{N}{\eta}} E_i$.
\end{proposition}

\begin{proof}
Using (\ref{L_n}),
\begin{align*}
\Psi_m =& L_{-m} (e^{F[E,0]}g^{-1})=
\sum_{i=1}^{N}\left(\alpha(a_{i}^{\dagger})^{m+1}a_{i}+(1-\alpha)a_{i}(a_{i}^{\dagger})^{m+1}\right)e^{F[E,0]}g^{-1}\notag\\
=&\sum_{i=1}^{N}\left(a_{i}(a_{i}^{\dagger})^{m+1}-\alpha(m+1)(a_{i}^{\dagger})^{m}\right)e^{F[E,0]}g^{-1} \notag \\
=&
\left(\frac{\eta}{N}\right)^{\frac{N(N-1)}{4}}e^{\frac{1}{2}\sum_{k=1}^{N}y_{k}^{2}}\left(-\frac{1}{\sqrt{2}}\right)^{m}\notag\\
& \Biggl\{ - \sum_{l=1}^{N}y_{l}\left(\frac{\partial}{\partial y_{l}}\right)^{m+1} 
-\frac{1}{2}\sum_{l=1}^{N}\left(\frac{\partial}{\partial y_{l}}\right)^{m+2} 
- \alpha(m+1)\sum_{l=1}^{N}\left(\frac{\partial}{\partial y_{l}}\right)^{m}  \Biggl\} \notag\\
& e^{F[E,0]-\sum_{k=1}^{N}y_{k}^{2}+\sum_{1\leq i<j\leq N}\log\left(y_{j}-y_{i}\right)} .
\end{align*}
To derive the third equality, the following formula was used:
\begin{align}
a_{i}^{\dagger}=&\frac{1}{\sqrt{2}}\left(y_{i}-\frac{\partial}{\partial y_{i}}\right)=-\frac{1}{\sqrt{2}}e^{\frac{1}{2}\sum_{k=1}^{N}y_{k}^{2}}\frac{\partial}{\partial y_{i}}e^{-\frac{1}{2}\sum_{k=1}^{N}y_{k}^{2}} .
\end{align}
\end{proof}

The following proposition requires a special case of the Fa\`a di Bruno's formula
\begin{align}
  \left\{e^{f(z)}\right\}^{(n)}=e^{f(z)}
  \sum_{r=0}^{n}B_{n,r}(f^{(1)},f^{(2)},\cdots,f^{(n-r+1)}) \label{bell_formula}
\end{align}
in terms of the Bell polynomials
\begin{align*}
  &B_{n,r}(f_{1},f_{2},\cdots,f_{n-r+1})
  \\
  &=
\sum_{\substack{  j_{1}+j_{2}+\cdots+j_{n-r+1}=r,\\
    j_{1}+2j_{2}+\cdots+(n-r+1)j_{n-r+1}=n}}
\hspace*{-1em}
\frac{n!}{j_{1}!j_{2}!\cdots j_{n-r+1}!}\left(\frac{f_{1}}{1!}\right)^{j_{1}}\left(\frac{f_{2}}{2!}\right)^{j_{2}}\cdots\left(\frac{f_{n-r+1}}{(n-r+1)!}\right)^{j_{n-r+1}}.
\end{align*}
We adopt the slightly extended Bell polynomials with $B_{0,0}=1$ and
$B_{n, 0}=0~ (n>0)$ so that the equality holds even when $n=0$.

From Proposition \ref{Prop3_3} and (\ref{bell_formula}),
we obtain the following expression of the
eigenfunction $\Psi_m$ of $\mathcal{H}_{HO}$.
\begin{proposition}
  The eigenfunction $\Psi_m := L_{-m} \Psi ~ (m \ge 0)$ is expressed by
  using Bell polynomials as follows:
\begin{align}
\Psi_m
=& - \left(-\frac{1}{\sqrt{2}}\right)^{m} \Psi ~
\sum_{l=1}^N \Big\{
y_l \sum_{r=0}^{m+1} B_{m+1 , r}^l +\frac{1}{2}   \sum_{r=0}^{m+2} B_{m+2 , r}^l  
+ \alpha (m+1) \sum_{r=0}^{m} B_{m , r}^l 
\Big\} \label{formula_higher_eigen}
\end{align}
where 
$$ 
 B_{m , r}^l := B_{m,r}\Biggl( \left(\frac{\partial}{\partial y_{l}}\right)^{1}(F[E,0]-V[E]),\cdots,\left(\frac{\partial}{\partial y_{l}}\right)^{m-r+1}(F[E,0]-V[E])\Biggl) ,
$$
and $$V[E]= \sum_{k=1}^{N}y_{k}^{2}-\!\!\! \sum_{1\leq i<j\leq N} \!\! \log(y_{j}-y_{i})
= \frac{N}{\eta}  \sum_{k=1}^{N}E_{k}^{2}- \!\!\! \sum_{1\leq i<j\leq N} \!\! \log(E_{j}-E_{i})
-\frac{N(N-1)}{4} \log \frac{N}{\eta}.$$
\end{proposition}


\section{Loop  Equation}\label{sect4}


Let us consider more general Schwinger-Dyson equations called loop equations.
For details on loop equations and related calculations, refer to \cite{Eynard:2015aea}.

\begin{align*}
I_{k_{1},k_{2},\cdots,k_{n}}:=&\sum_{i,j=1}^{N}\frac{1}{\mathcal{Z}[E,0]}\int_{H(N)}d\Phi\frac{\partial}{\partial\Phi_{ij}}
\left\{ \left(\Phi^{k_{1}}\right)_{ij}\mathrm{Tr}\left[\Phi^{k_{2}}\right]\cdots\mathrm{Tr}\left[\Phi^{k_{n}}\right]e^{-S[\Phi]} \right\} =0,
\end{align*}
where $\frac{\partial}{\partial\Phi_{ij}}=\frac{1}{2}\left(\frac{\partial}{\partial\Phi^{\mathrm{Re}}_{ij}}-i\frac{\partial}{\partial\Phi^{\mathrm{Im}}_{ij}}\right)$ for $i\neq j$.
Recall the following useful formulae  
\begin{align}
\sum_{i,j=1}^{N}\frac{\partial}{\partial\Phi_{ij}}\left(\Phi^{k}\right)_{ij}=
\left\{ \,
    \begin{aligned}
    & 0 \hspace{2mm}(k=0)\\
    & N^{2} \hspace{2mm}(k=1)\\
    & \sum_{l=0}^{k-1}\mathrm{Tr}\Phi^{l}\mathrm{Tr}\Phi^{k-1-l} \hspace{2mm}(k>1)
    \end{aligned}
\right.\label{b}
\end{align}
\begin{align}
\sum_{i,j=1}^{N}\left(\Phi^{k_{1}}\right)_{ij}\frac{\partial}{\partial\Phi_{ij}}\mathrm{Tr}\Phi^{k}=&\sum_{i,j=1}^{N}\left(\Phi^{k_{1}}\right)_{ij}k\left(\Phi^{k-1}\right)_{ji}=k\mathrm{Tr}\Phi^{k+k_{1}-1} . \label{c}
\end{align}
In addition, using the computation 
\begin{align}
\sum_{i,j=1}^{N}\left(\Phi^{k_{1}}\right)_{ij}\frac{\partial}{\partial\Phi_{ij}}e^{-S[\Phi]}=&-N\sum_{i,j=1}^{N}\left(\Phi^{k_{1}}\right)_{ij}\frac{\partial}{\partial\Phi_{ij}}\left(\mathrm{Tr}E\Phi^{2}+\frac{\eta}{4}\mathrm{Tr}\Phi^{4}\right)e^{-S[\Phi]}\notag\\
=&-N\mathrm{Tr}\left(2E\Phi+\eta\Phi^{3}\right)\Phi^{k_{1}}e^{-S[\Phi]} , 
\end{align}
and definition $L'\left(\Phi,E\right):= \left(2E\Phi+\eta\Phi^{3}\right)$,
we obtain 
\begin{align}
\sum_{i,j=1}^{N}\left(\Phi^{k_{1}}\right)_{ij}\frac{\partial}{\partial\Phi_{ij}}e^{-S[\Phi]}=-N\mathrm{Tr}\left(L'\left(\Phi,E\right)\Phi^{k_{1}}\right)e^{-S[\Phi]}.  \label{a}
\end{align}


From (\ref{a}),(\ref{b}),(\ref{c}), the loop equations 
\begin{align}
&N\langle\mathrm{Tr}(\Phi^{k_{2}})\cdots\mathrm{Tr}(\Phi^{k_{n}})\mathrm{Tr}\left(L'(\Phi,E)\Phi^{k_{1}}\right)\rangle
\notag\\
=&(1-\delta_{k_{1}0})\sum_{l=0}^{k_{1}-1}\langle\mathrm{Tr}(\Phi^{l})\mathrm{Tr}\Phi^{k_{1}-1-l}\mathrm{Tr}(\Phi^{k_{2}})\cdots\mathrm{Tr}(\Phi^{k_{n}})\rangle\notag\\
&+\sum_{l=2}^{n} k_{l}\langle\mathrm{Tr}(\Phi^{k_{2}})\cdots\mathrm{Tr}(\Phi^{k_{l-1}})\left(\mathrm{Tr}\Phi^{k_{l}+k_{1}-1}\right)\mathrm{Tr}(\Phi^{k_{l+1}})\cdots\mathrm{Tr}(\Phi^{k_{n}})\rangle \label{loop1}
\end{align}
are obtained for arbitaraly non-negative integers $k_1 , k_2 , \cdots , k_n$.

We introduce resolvents by
\begin{align*}
R(u):=R^{Id}(u) :=&\mathrm{Tr}\left(\frac{1}{u-\Phi}\right)
=\frac{1}{u}\sum_{k=0}^{\infty}\frac{1}{u^{k}}\mathrm{Tr}\Phi^{k}, \qquad R^E(u):=\mathrm{Tr}\left(\frac{E}{u-\Phi}\right) .
\end{align*}
After multiplying equation (\ref{loop1}) by $1/u^{k_1 +1} , 1/u_2^{k_2+1} , \cdots , 1/u_n^{k_n+1}$,
taking sum for each $k_1, \cdots, k_n$, then 
we get the loop equation.
\begin{align}
&N\left\langle\mathrm{Tr}\left(L'(\Phi,E)\frac{1}{u-\Phi}\right)R(u_{2})\cdots R(u_{n})\right\rangle\notag\\
=&\left\langle R^{2}(u)R(u_{2})R(u_{3})\cdots R(u_{n})\right\rangle\notag\\
&+\sum_{l=2}^{n}\frac{\partial}{\partial u_{l}}\left(\frac{\left\langle R(u)R(u_{2})\cdots R(u_{l-1})R(u_{l+1})\cdots R(u_{n})\right\rangle
-\left\langle R(u_{2})\cdots R(u_{n})\right\rangle}{u-u_{l}}\right) . \label{j}
\end{align}

We introduce ${\mathcal U}=\{u_{2},u_{3},\cdots,u_{n}\}$,
\begin{align}
\hat{P}_{n-1}(u,{\mathcal U}):=\left\langle\mathrm{Tr}\left(\left(L'(u,E)-L'(\Phi,E)\right)\frac{1}{u-\Phi}\right)R(u_{2})\cdots R(u_{n})\right\rangle,
\end{align}
and
\begin{align}
\hat{R}_{n}^A(u_{1},\cdots,u_{n})=&\left\langle\mathrm{Tr}\frac{A}{u_{1}-\Phi}  \mathrm{Tr}\frac{1}{u_{2}-\Phi}   \cdots\mathrm{Tr}\frac{1}{u_{n}-\Phi}\right\rangle
,
\end{align}
for $A \in \{ Id,~ E \}$.  In particular, when $A=Id$, we simply write $\hat{R}_{n}(u_{1},\cdots,u_{n})$:
$$
\hat{R}_{n}(u_{1},\cdots,u_{n})=
\hat{R}_{n}^{Id}(u_{1},\cdots,u_{n}) =\left\langle R(u_{1})\cdots R(u_{n})\right\rangle .
$$
Using these symbols, (\ref{j}) is rewritten as
%
%
\begin{align}
N\left(
 \eta u^3\widehat R_n(u,\mathcal{U})
 +
 2u\widehat R_n^E(u,\mathcal{U})
 -
 \widehat P_{n-1}(u,\mathcal{U})
 \right)
=&\hat{R}_{n+1}(u,u,{\mathcal U})+\sum_{l=2}^{n}\frac{\partial}{\partial u_{l}}\frac{\hat{R}_{n-1}(u,{\mathcal U}\backslash\{u_{l}\})-\hat{R}_{n-1}({\mathcal U})}{u-u_{l}} . \label{loop_2}
\end{align}

We define the multipoint cumulant resolvent,
\begin{align}
R_{n}^A (u_{1},\cdots,u_{n})=&\left\langle\mathrm{Tr}\frac{A}{u_{1}-\Phi} \mathrm{Tr}\frac{1}{u_{2}-\Phi}\cdots\mathrm{Tr}\frac{1}{u_{n}-\Phi}\right\rangle_{c}
= \left\langle R^A(u_{1}) R(u_{1})\cdots R(u_{n})\right\rangle_{c} ,\notag\\
R_{n}(u_{1},\cdots,u_{n})=&\left\langle\mathrm{Tr}\frac{1}{u_{1}-\Phi}\cdots\mathrm{Tr}\frac{1}{u_{n}-\Phi}\right\rangle_{c}
= \left\langle R(u_{1})\cdots R(u_{n})\right\rangle_{c} ,
\end{align}
where
\begin{align}
R_{n}^A(u_{1}.\cdots,u_{n}):=&\left.\partial_{s_{1}}\cdots\partial_{s_{n}}\log\left\langle\exp\left( s_1 R^A(u_1) +\sum_{i=2}^{N}s_{i}R(u_{i})\right)\right\rangle\right|_{s_{1}=\cdots=s_{n}=0}
\notag\\
R_{n}(u_{1}.\cdots,u_{n}):=&\left.\partial_{s_{1}}\cdots\partial_{s_{n}}\log\left\langle\exp\left(\sum_{i=1}^{N}s_{i}R(u_{i})\right)\right\rangle\right|_{s_{1}=\cdots=s_{n}=0}.
\end{align}
Similarly ${P}_{n-1}(u,{\mathcal U})$ is defined by
\begin{align}
{P}_{n-1}(u,{\mathcal U}):=\left\langle\mathrm{Tr}\left(\left(L'(u,E)-L'(\Phi,E)\right)\frac{1}{u-\Phi}\right)R(u_{2})\cdots R(u_{n})\right\rangle_c .
\end{align}

The above discussions give in 
the simplest case $n=1$ (we put $k_2= k_3 = \cdots = k_n=0$
from the begining of Section \ref{sect4})
\begin{align}
N\left\langle\mathrm{Tr}\left(L'(\Phi,E)\frac{1}{u-\Phi}\right) \right\rangle
=\left\langle R^{2}(u)\right\rangle.
\end{align}
This equation is rewritten in terms of cumulants as
\begin{align}
N\left\langle\mathrm{Tr}\left(L'(\Phi,E)\frac{1}{u-\Phi}\right) \right\rangle_c
=\left\langle R^{2}(u)\right\rangle_c + \left\langle R (u)\right\rangle_c \left\langle R (u)\right\rangle_c .
\label{cum_n=1}
\end{align}

Similarly, (\ref{j}) or (\ref{loop_2}) is rewritten as follows:
\begin{proposition}
\begin{align}
&
N\left(  \eta u^3R_n(u,\mathcal{U}) +  2uR_n^E(u,\mathcal{U})  -  P_{n-1}(u,\mathcal{U})  \right)
= R_{n+1}(u,u,{\mathcal U})
\notag \\
&+ \!\!\!\! \sum_{J(\vec{k})\amalg  J(\vec{l})={\mathcal U}} \!\!\!\! 
R_{|\vec{k}|+1}(u,J(\vec{k}))R_{_{|\vec{l}|+1}}(u,J(\vec{l}))
+\sum_{u_l \in{\mathcal U}}\frac{\partial}{\partial u_{l}}\frac{R_{n-1}(u,{\mathcal U} \backslash\{u_{l}\})-R_{n-1}({\mathcal U})}{u-u_{l}} . \label{loop_cumulant}
\end{align}
Here ${\mathcal U} = \{u_2, \cdots, u_n \}$ and 
$J(\vec{k}):= J_{k_1,\cdots, k_j} := \{ u_{k_1} , \cdots , u_{k_j} \} \subset  {\mathcal U} $.
\end{proposition}
The result is almost the same as the loop equation of the normal Hermitian
one-matrix model \cite{Eynard:2015aea}, 
despite the fact that $E$ is attached to the kinetic term of the matrix model we are considering.
However, as we have been unable to find a paper that rigorously
details the process of rewriting it in terms of cumulants, 
we provide a proof of this for the convenience of the reader.

\begin{proof}
$n=2$ case of (\ref{loop_2}) is given by
\begin{align}
N\left\langle\mathrm{Tr}\left(L'(\Phi,E)\frac{1}{u-\Phi}\right)R(u_{2}) \right\rangle
=\left\langle R^{2}(u)R(u_{2}) \right\rangle
+\frac{\partial}{\partial u_{2}}\left(\frac{\left\langle R(u)\right\rangle-\left\langle R(u_{2})\right\rangle}{u-u_{2}}\right)
\end{align}
By definition of cumulants,
this is rewritten as
\begin{align}
&N\left( \left\langle\mathrm{Tr}\left(L'(\Phi,E)\frac{1}{u-\Phi}\right)R(u_{2}) \right\rangle_c
+ \left\langle\mathrm{Tr}\left(L'(\Phi,E)\frac{1}{u-\Phi}\right)\right\rangle_c  \left\langle R(u_{2}) \right\rangle_c \right)
=
\notag \\
&\left\langle R^{2}(u)R(u_{2}) \right\rangle_ c
+\left\langle R^{2}(u)\right\rangle_ c \left\langle R(u_{2}) \right\rangle_ c
+\left\langle R(u)\right\rangle_ c^2 \left\langle R(u_{2}) \right\rangle_ c
+2 \left\langle R(u) \right\rangle_ c \left\langle R(u) R(u_{2}) \right\rangle_ c
\notag \\
&+\frac{\partial}{\partial u_{2}}\left(\frac{\left\langle R(u)\right\rangle_c-\left\langle R(u_{2})\right\rangle_c }{u-u_{2}}\right)
.
\end{align}
Using (\ref{cum_n=1}), this is reduced as
\begin{align}
&N \left\langle\mathrm{Tr}\left(L'(\Phi,E)\frac{1}{u-\Phi}\right)R(u_{2}) \right\rangle_c
=
\notag \\
&\left\langle R^{2}(u)R(u_{2}) \right\rangle_ c
+2 \left\langle R(u) \right\rangle_ c \left\langle R(u) R(u_{2}) \right\rangle_ c
+\frac{\partial}{\partial u_{2}}\left(\frac{\left\langle R(u)\right\rangle_c-\left\langle R(u_{2})\right\rangle_c }{u-u_{2}}\right)
.
\end{align}
This equation is the case of $n=2$ in (\ref{loop_cumulant}).

We assume the above equations (\ref{loop_cumulant}) for $n=2, 3, \cdots , n$  to be satisfied.
As we proved it above, the $n+1$ case of (\ref{loop_2}) is satisfied. 
\begin{align}
N\left( \eta u^3\widehat R_{n+1}(u, \mathcal{U}')
 +  2u\widehat R_{n+1}^E(u,\mathcal{U}') -
 \widehat P_{n}(u,\mathcal{U}')
 \right)
=&\hat{R}_{n+2}(u,u,{\mathcal U}')+\sum_{l=2}^{n+1}\frac{\partial}{\partial u_{l}}\frac{\hat{R}_{n}(u,{\mathcal U}'\backslash\{u_{l}\})-\hat{R}_{n}({\mathcal U}')}{u-u_{l}} , \label{n+1start}
\end{align}
where ${\mathcal U}'=\{u_{2},\cdots,u_{n+1}\}$.

Note that ${\mathcal U}'=\{u_{2},\cdots,u_{n+1}\}=\left({\mathcal U} \backslash J_{k_{1}\cdots k_{j}}\right)\amalg\left(\{u_{n+1}\}\cup J_{k_{1}\cdots k_{j}}\right)$ and $J_{k_{1}\cdots k_{j}}=\{u_{k_{1}},\cdots,u_{k_{j}} ~| ~2\leq k_{j}\leq n\} \subset {\mathcal U}$.
Then the following is obtained:
\begin{align}
\hat{R}_{n}({\mathcal U}')=&\sum_{j=0}^{n-1}\sum_{J_{k_{1}\cdots k_{j}}\subset {\mathcal U}}\hat{R}_{n-j-1}\left({\mathcal U}\backslash J_{k_{1}\cdots k_{j}}\right)R_{j+1}\left(J_{k_{1}\cdots k_{j}}\cup\{u_{n+1}\}\right) .
\end{align}

Similarly, we use the following identities.
(In the following, it appears that formula numbers have been assigned excessively. The purpose is to assign a number to each term on the right side of each formula and to clearly indicate which term is being referred to by the number.)
\begin{align}
\hat{R}_{n+1}^A (u,{\mathcal U}')=&\sum_{j=0}^{n-1}\sum_{J_{k_{1}\cdots k_{j}}\subset{\mathcal U}}\hat{R}_{n-j}^A (u,{\mathcal U}\backslash J_{k_{1}\cdots k_{j}})R_{j+1}\left(J_{k_{1}\cdots k_{j}}\cup\{u_{n+1}\}\right)\label{ba}\\
+&\sum_{j=0}^{n-1}\sum_{J_{k_{1}\cdots k_{j}}\subset{\mathcal U}}\hat{R}_{n-j-1}\left({\mathcal U}\backslash J_{k_{1}\cdots k_{j}}\right)R_{j+2}^A\left(\{u,u_{n+1}\}\cup J_{k_{1}\cdots k_{j}}\right),
\end{align}
\begin{align}
\hat{P}_{n}(u,{\mathcal U}'):=&\left\langle\mathrm{Tr}\left(\left(L'(u,E)-L'(\Phi,E)\right)\frac{1}{u-\Phi}\right)R(u_{2})\cdots R(u_{n+1})\right\rangle\notag\\
=&\sum_{j=0}^{n-1}\sum_{J_{k_{1}\cdots k_{j}}\subset{\mathcal U}}\hat{P}_{n-j-1}(u,{\mathcal U}\backslash J_{k_{1}\cdots k_{j}})R_{j+1}\left(J_{k_{1}\cdots k_{j}}\cup\{u_{n+1}\}\right)\label{bb}\\
&+\sum_{j=0}^{n-1}\sum_{J_{k_{1}\cdots k_{j}}\subset{\mathcal U}}\hat{R}_{n-1-j}\left({\mathcal U}\backslash J_{k_{1}\cdots k_{j}}\right)P_{j+1}(u,J_{k_{1}\cdots k_{j}},u_{n+1}),
\end{align}
\begin{align}
\hat{R}_{n+2}(u,u,{\mathcal U}')=&\sum_{j=0}^{n-1}\sum_{J_{k_{1}\cdots k_{j}}\subset{\mathcal U}}\hat{R}_{n+1-j}\left(u,u,{\mathcal U}\backslash J_{k_{1}\cdots k_{j}}\right)R_{j+1}\left(J_{k_{1}\cdots k_{j}}\cup\{u_{n+1}\}\right)\label{bc}\\
&+2\sum_{j=0}^{n-1}\sum_{J_{k_{1}\cdots k_{j}}\subset{\mathcal U}}\hat{R}_{n-j}\left(u,{\mathcal U}\backslash J_{k_{1}\cdots k_{j}}\right)R_{j+2}\left(J_{k_{1}\cdots k_{j}},u,u_{n+1}\right)\\
&+\sum_{j=0}^{n-1}\sum_{J_{k_{1}\cdots k_{j}}\subset{\mathcal U}}\hat{R}_{n-1-j}\left({\mathcal U}\backslash J_{k_{1}\cdots k_{j}}\right)R_{j+3}\left(J_{k_{1}\cdots k_{j}},u,u,u_{n+1}\right),
\end{align}
\begin{align}
&\sum_{l=2}^{n+1}\frac{\partial}{\partial u_{l}}\frac{\hat{R}_{n}\left(u,{\mathcal U}'\backslash\{u_{l}\}\right)}{u-u_{l}}=\left.\sum_{l=2}^{n+1}\hat{R}_{n}({\mathcal U}')\right|_{u_{l}=u}\frac{\partial}{\partial u_{l}}\frac{1}{u-u_{l}}\notag\\
=&\left.\sum_{l=2}^{n+1}\sum_{j=0}^{n-1}\sum_{J_{k_{1}\cdots k_{j}}\subset{\mathcal U}}\hat{R}_{n-j-1}\left({\mathcal U}\backslash J_{k_{1}\cdots k_{j}}\right)R_{j+1}(u_{n+1},J_{k_{1}\cdots k_{j}})\right|_{u_{l}=u}\frac{\partial}{\partial u_{l}}\frac{1}{u-u_{l}}\notag\\
=&\sum_{j=0}^{n-1}\sum_{J_{k_{1}\cdots k_{j}}\subset{\mathcal U}}\Biggl\{\left.\sum_{u_l \in{\mathcal U}\backslash J_{k_{1}\cdots k_{j}}}\hat{R}_{n-j-1}\left({\mathcal U}\backslash J_{k_{1}\cdots k_{j}}\right)\right|_{u_{l}=u}R_{j+1}\left(u_{n+1},J_{k_{1}\cdots k_{j}}\right)\frac{\partial}{\partial u_{l}}\frac{1}{u-u_{l}}\label{bd}\\
&+\left.\sum_{u_l \in\{u_{n+1}\}\cup J_{k_{1}\cdots k_{j}}}\hat{R}_{n-j-1}\left({\mathcal U}\backslash J_{k_{1}\cdots k_{j}}\right)R_{j+1}(u_{n+1},J_{k_{1}\cdots k_{j}})\right|_{u_{l}=u}\frac{\partial}{\partial u_{l}}\frac{1}{u-u_{l}}\Biggl\} ,
\end{align}
and 
\begin{align}
&\sum_{l=2}^{n+1}\frac{\partial}{\partial u_{l}}\frac{\hat{R}_{n}({\mathcal U}')}{u-u_{l}}=\sum_{l=2}^{n+1}\frac{\partial}{\partial u_{l}}\left(\frac{1}{u-u_{l}}\sum_{j=0}^{n-1}\sum_{J_{k_{1}\cdots k_{j}}\subset{\mathcal U}}\hat{R}_{n-j-1}({\mathcal U}\backslash J_{k_{1}\cdots k_{j}})R_{j+1}\left(J_{k_{1}\cdots k_{j}}\cup\{u_{n+1}\}\right)\right)\notag\\
&=\sum_{j=0}^{n-1}\sum_{J_{k_{1}\cdots k_{j}}\subset{\mathcal U}}\Biggl\{\sum_{u_l \in{\mathcal U}\backslash J_{k_{1}\cdots k_{j}}}\frac{\partial}{\partial u_{l}}\left(\frac{1}{u-u_{l}}\hat{R}_{n-j-1}({\mathcal U}\backslash J_{k_{1}\cdots k_{j}})\right)R_{j+1}\left(J_{k_{1}\cdots k_{j}}\cup\{u_{n+1}\}\right) \notag \\
&+\hat{R}_{n-j-1}\left({\mathcal U}\backslash J_{k_{1}\cdots k_{j}}\right)\sum_{u_l \in\{u_{n+1}\}\cup J_{k_{1}\cdots k_{j}}}\frac{\partial}{\partial u_{l}}\frac{1}{u-u_{l}}R_{j+1}\left(J_{k_{1}\cdots k_{j}}\cup\{u_{n+1}\}\right)\Biggl\} \label{be}
\end{align}
are obtained.
In what follows, we use the notation
$(\eta u^3 \Big|_{A=Id} + 2u \Big|_{A=E})$
in the sense that, for example,
\[
\left(\eta u^3 \Big|_{A=Id} + 2u \Big|_{A=E}\right) R_n^A ( \cdots  )
= \eta u^3 R_n ( \cdots )
+ 2u R_n^E ( \cdots  ).
\]

By the assumption,
\begin{align}
N \left(\eta u^3 \Big|_{A=Id} + 2u \Big|_{A=E}\right) (\ref{ba})
-N \times (\ref{bb})-(\ref{bc})-(\ref{bd})+(\ref{be})=0,
\end{align}
then (\ref{n+1start})  are rewritten as
\begin{align}
&N\Biggl(\left(\eta u^3 \Big|_{A=Id} + 2u \Big|_{A=E}\right)\sum_{j=0}^{n-1}\sum_{J_{k_{1}\cdots k_{j}}\subset{\mathcal U}}\hat{R}_{n-j-1}\left({\mathcal U}\backslash J_{k_{1}\cdots k_{j}}\right)R_{j+2}^A \left(\{u,u_{n+1}\}\cup J_{k_{1}\cdots k_{j}}\right)\notag\\
-&\sum_{j=0}^{n-1}\sum_{J_{k_{1}\cdots k_{j}}\subset{\mathcal U}}\hat{R}_{n-1-j}\left({\mathcal U}\backslash J_{k_{1}\cdots k_{j}}\right)P_{j+1}(u,J_{k_{1}\cdots k_{j}},u_{n+1})\Biggl)\notag\\
-&2\sum_{j=0}^{n-1}\sum_{J_{k_{1}\cdots k_{j}}\subset{\mathcal U}}\hat{R}_{n-j}\left(u,{\mathcal U}\backslash J_{k_{1}\cdots k_{j}}\right)R_{j+2}\left(J_{k_{1}\cdots k_{j}},u,u_{n+1}\right)\notag\\
-&\sum_{j=0}^{n-1}\sum_{J_{k_{1}\cdots k_{j}}\subset{\mathcal U}}\hat{R}_{n-1-j}\left({\mathcal U}\backslash J_{k_{1}\cdots k_{j}}\right)R_{j+3}\left(J_{k_{1}\cdots k_{j}},u,u,u_{n+1}\right)\notag\\
-&\sum_{j=0}^{n-1}\sum_{J_{k_{1}\cdots k_{j}}\subset{\mathcal U}}\Biggl\{\left.\sum_{u_l \in\{u_{n+1}\}\cup J_{k_{1}\cdots k_{j}}}\hat{R}_{n-j-1}\left({\mathcal U}\backslash J_{k_{1}\cdots k_{j}}\right)R_{j+1}(u_{n+1},J_{k_{1}\cdots k_{j}})\right|_{u_{l}=u}\frac{\partial}{\partial u_{l}}\frac{1}{u-u_{l}}\notag\\
-&\hat{R}_{n-j-1}\left({\mathcal U}\backslash J_{k_{1}\cdots k_{j}}\right)\sum_{u_l \in\{u_{n+1}\}\cup J_{k_{1}\cdots k_{j}}}\frac{\partial}{\partial u_{l}}\frac{1}{u-u_{l}}R_{j+1}\left(J_{k_{1}\cdots k_{j}}\cup\{u_{n+1}\}\right)\Biggl\} =0.
\end{align}
Expanding the third line further into cumulants, we obtain
\begin{align}
0=&\sum_{j=0}^{n-1}\sum_{J_{k_{1}\cdots k_{j}}\subset{\mathcal U}}\hat{R}_{n-j-1}({\mathcal U}\backslash J_{k_{1}\cdots k_{j}})\Biggl[N
\left(\left(\eta u^3 \Big|_{A=Id} \!\!\!\!\!\!\!+ 2u \Big|_{A=E}\right)R_{j+2}^A(u,u_{n+1},J_{k_{1}\cdots k_{j}})-P_{j+1}(u,u_{n+1},J_{k_{1}\cdots k_{j}})\right)\notag\\
&-\Biggl\{\left.\sum_{u_l \in\{u_{n+1}\}\cup J_{k_{1}\cdots k_{j}}}R_{j+1}(u_{n+1},J_{k_{1}\cdots k_{j}})\right|_{u_{l}=u}\frac{\partial}{\partial u_{l}}\frac{1}{u-u_{l}}\notag\\
&-\sum_{u_l \in\{u_{n+1}\}\cup J_{k_{1}\cdots k_{j}}}\frac{\partial}{\partial u_{l}}\frac{1}{u-u_{l}}R_{j+1}\left(J_{k_{1}\cdots k_{j}}\cup\{u_{n+1}\}\right)\Biggl\}-R_{j+3}\left(J_{k_{1}\cdots k_{j}},u,u,u_{n+1}\right)\Biggl]\notag\\
-&\sum_{j=0}^{n-1} \!
\sum_{J_{k_{1}\cdots k_{j}}\subset{\mathcal U}}\Biggl\{2\sum_{i=0}^{n-1-j}
\!\!\!\!
\sum_{J(\vec{l}_{i})\subset\left({\mathcal U}-J(\vec{k}_{j})\right)}
\!\!\!\!\! \!\!\!
\hat{R}_{n-j-i -1}({\mathcal U}-J(\vec{k}_{j})-J(\vec{l}_{i}))
R_{i+1}( {\mathcal U}-J(\vec{l}_{i}) )
R_{j+2}(u,u_{n+1},J_{k_{1}\cdots k_{j}}) \Biggl\},\label{bbb}
\end{align}
where $J(\vec{k}_{j})=J_{k_{1}\cdots k_{j}}$. 
We transform the last line in (\ref{bbb}).
\begin{align}
&-2\sum_{j=0}^{n-1}\sum_{J_{k_{1}\cdots k_{j}}\subset{\mathcal U}}\sum_{i=0}^{n-1-j}
\!\!\!
\sum_{J(\vec{l}_{i})\subset\left({\mathcal U}-J(\vec{k}_{j})\right)}
\!\!\!\!\! \!\!\!
\hat{R}_{n-j-i -1}({\mathcal U}-J(\vec{k}_{j})-J(\vec{l}_{i}))
R_{i+1}( {\mathcal U}-J(\vec{l}_{i}) )
R_{j+2}(u,u_{n+1},J_{k_{1}\cdots k_{j}}) 
\notag\\
=&-2\sum_{j=0}^{n-1}\sum_{i=0}^{n-1-j}\sum_{J(\vec{k}_{j})\subset{\mathcal U}}
\sum_{J(\vec{k}_{j}+\vec{l}_{i})\subset{\mathcal U}}
\!\!\!\!\! \!\!\!
\hat{R}_{n-j-i -1}({\mathcal U}-J(\vec{k}_{j})-J(\vec{l}_{i}))
R_{i+1}( {\mathcal U}-J(\vec{l}_{i}) )
R_{j+2}(u,u_{n+1},J_{k_{1}\cdots k_{j}}) \notag\\
=&-2\sum_{m=0}^{n-1}\sum_{J(\vec{k}_{m})\subset{\mathcal U}}\sum_{i+j=m}\sum_{J(k_{j})\amalg J(\vec{l}_{i})=J(\vec{k}_{m})}
\!\!\!\!\! \!\!\!\!\!\!
\hat{R}_{n-j-i -1}({\mathcal U}-J(\vec{k}_{j})-J(\vec{l}_{i}))
R_{i+1}( {\mathcal U}-J(\vec{l}_{i}) )
R_{j+2}(u,u_{n+1},J_{k_{1}\cdots k_{j}}) \notag\\
=&-2\sum_{m=0}^{n-1}\sum_{J(\vec{k}_{m})\subset{\mathcal U}}\hat{R}_{n-1-m}\left({\mathcal U}\backslash J(\vec{k}_{m})\right)\sum_{i+j=m}\sum_{J(\vec{k}_{i})\amalg J(\vec{k}_{i})=J(\vec{k}_{m})}
\!\!\!\!\! \!\!\!
R_{i+1}(u,J(\vec{k}_{i}))R_{j+2}\left(u,u_{n+1},J(\vec{k}_{j})\right)\notag\\
=&-\sum_{m=0}^{n-1}\sum_{J(\vec{k}_{m})\subset{\mathcal U}}\hat{R}_{n-1-m}\left({\mathcal U}\backslash J(\vec{k}_{m})\right)\sum_{i+j=m+1}\sum_{J(\vec{k}_{i})\amalg  J(\vec{k}_{j})=J(\vec{k}_{m})\cup\{u_{n+1}\}}
\!\!\! \!\!\!
R_{i+1}(u,J(\vec{k}_{i}))R_{j+1}(u,J(\vec{k}_{j})) . \label{c2_deform}
\end{align}
Substituting (\ref{c2_deform}) in the last line in  (\ref{bbb}),
the following is obtained.
\begin{align}
0=&\sum_{j=0}^{n-1}\sum_{J_{k_{1}\cdots k_{j}}\subset{\mathcal U}}\hat{R}_{n-j-1}({\mathcal U}\backslash J_{k_{1}\cdots k_{j}})\Biggl[N\left(
\left(\eta u^3 \Big|_{A=Id} \!\!\!\!\!\!\!+ 2u \Big|_{A=E}\right)R_{j+2}^A(u,u_{n+1},J_{k_{1}\cdots k_{j}})
-P_{j+1}(u,u_{n+1},J_{k_{1}\cdots k_{j}})\right)\notag\\
&-\Biggl\{\left.\sum_{u_l \in\{u_{n+1}\}\cup J_{k_{1}\cdots k_{j}}}R_{j+1}(u_{n+1},J_{k_{1}\cdots k_{j}})\right|_{u_{l}=u}\frac{\partial}{\partial u_{l}}\frac{1}{u-u_{l}}\notag\\
&-\sum_{u_l \in\{u_{n+1}\}\cup J_{k_{1}\cdots k_{j}}}\frac{\partial}{\partial u_{l}}\frac{1}{u-u_{l}}R_{j+1}\left(J_{k_{1}\cdots k_{j}}\cup\{u_{n+1}\}\right)\Biggl\}-R_{j+3}\left(J_{k_{1}\cdots k_{j}},u,u,u_{n+1}\right)\Biggl]\notag\\
&-\sum_{m=0}^{n-1}\sum_{J(\vec{k}_{m})\subset{\mathcal U}}\hat{R}_{n-1-m}\left({\mathcal U}\backslash J(\vec{k}_{m})\right)\sum_{i+j=m+1}\sum_{J(\vec{k}_{i})\amalg  J(\vec{k}_{j})=J(\vec{k}_{m})\cup\{u_{n+1}\}}
\!\!\! \!\!\!\!\!\! \!\!\!
R_{i+1}(u,J(\vec{k}_{i}))R_{j+1}(u,J(\vec{k}_{j}))
\end{align}

When $J_{k_{1}\cdots k_{j}}\neq{\mathcal U}$, the coefficient of 
$\hat{R}_{n-j-1}({\mathcal U}\backslash J_{k_{1}\cdots k_{j}})$
is 
\begin{align}
&N\left(\left(\eta u^3 \Big|_{A=Id} \!\!\!+ 2u \Big|_{A=E}\right)R_{j+2}^A(u,u_{n+1},J_{k_{1}\cdots k_{j}})-P_{j+1}(u,u_{n+1},J_{k_{1}\cdots k_{j}})\right)\notag\\
&-\Biggl\{\left.\sum_{u_l \in\{u_{n+1}\}\cup J_{k_{1}\cdots k_{j}}}R_{j+1}(u_{n+1},J_{k_{1}\cdots k_{j}})\right|_{u_{l}=u}\frac{\partial}{\partial u_{l}}\frac{1}{u-u_{l}}\notag\\
&-\sum_{u_l \in\{u_{n+1}\}\cup J_{k_{1}\cdots k_{j}}}\frac{\partial}{\partial u_{l}}\frac{1}{u-u_{l}}R_{j+1}\left(J_{k_{1}\cdots k_{j}}\cup\{u_{n+1}\}\right)\Biggl\}-R_{j+3}\left(J_{k_{1}\cdots k_{j}},u,u,u_{n+1}\right)\notag\\
&-\sum_{i+l=j}\sum_{J(\vec{k}_{i})\amalg J(\vec{k}_{l})=J(\vec{k}_{j})\cup\{u_{n+1}\}}R_{l+1}(u,J(\vec{k}_{l}))=0, \label{aaa}
\end{align}
because the assumption for $j+1 \le n-1 $.
From (\ref{aaa}),  the case
$
J_{k_{1}\cdots k_{j}}={\mathcal U}
$
gives still survived term. The expression is 
\begin{align}
0=&N\left(\left(\eta u^3 \Big|_{A=Id} \!\!\!+ 2u \Big|_{A=E}\right)R_{n+1}^A(u,u_{n+1},{\mathcal U})-P_{n}(u,u_{n+1},{\mathcal U})\right)\notag\\
&-R_{n+2}(u,u,u_{n+1},{\mathcal U})-\sum_{J(\vec{k})\cup J(\vec{l})={\mathcal U}\cup\{u_{n+1}\}}R_{|\vec{k}|+1}(u,J(\vec{k}))R_{|\vec{l}|+1}(u,J(\vec{l}))\notag\\
&-\left\{\sum_{u_l \in\{u_{n+1}\}\cup{\mathcal U}}R_{n}(u_{n+1},{\mathcal U})|_{u_{l}=u}\frac{\partial}{\partial u_{l}}\frac{1}{u-u_{l}}-\sum_{u_l \in\{u_{n+1}\}\cup{\mathcal U}}\frac{\partial}{\partial u_{l}}\frac{1}{u-u_{l}}R_{n}({\mathcal U}\cup\{u_{n+1}\})\right\},
\end{align}
where $|\vec{k}|=i$ for $\vec{k} = \{ k_1, \cdots , k_i \}$.
This equation can be written as
\begin{align}
&0=N\left(\left(\eta u^3 \Big|_{A=Id} \!\!\!+ 2u \Big|_{A=E}\right)R_{n+1}^A(u,{\mathcal U}')-P_{n}(u,{\mathcal U}')\right)\notag\\
&-R_{n+2}(u,u,{\mathcal U}')-\sum_{J(\vec{k})\amalg  J(\vec{l})={\mathcal U}'}R_{|\vec{k}|+1}(u,J(\vec{k}))R_{_{|\vec{l}|+1}}(u,J(\vec{l}))-\sum_{u_l \in{\mathcal U}'}\frac{\partial}{\partial u_{l}}\frac{R_{n}(u,{\mathcal U}'\backslash\{u_{l}\})-R_{n}({\mathcal U}')}{u-u_{l}} .
\end{align}
It was thus proved by mathematical induction.
\end{proof}

The genus expansion of (\ref{loop_cumulant}) is given by
\begin{align}
0=&\left( \eta u^3 R_{g,n} (u,{\mathcal U}) + 2u R_{g,n}^E(u,{\mathcal U})-P_{g,n-1}(u,{\mathcal U})\right)\notag\\
&-R_{g-1,n+1}(u,u,{\mathcal U})-
\!\!\!\!\!
\sum_{J(\vec{k})\amalg  J(\vec{l})={\mathcal U}, h_{1}+h_{2}=g}
\!\!\!\!\!
R_{h_{1},|\vec{k}|+1}(u,J(\vec{k}))R_{h_{2},{|\vec{l}|+1}}(u,J(\vec{l}))\notag\\
&-\sum_{l=2}^{n}\frac{\partial}{\partial u_{l}}\frac{R_{g,n-1}(u,{\mathcal U}\backslash\{u_{l}\})-R_{g,n-1}({\mathcal U})}{u-u_{l}} . \label{Loop_genus}
\end{align}



\section{Schwinger-Dyson Equation for connected Green's function}
\label{sect5}

\subsection{Connected Green's function }\label{subset5_1}

Using $\displaystyle\log\frac{ \mathcal{Z}[E, J]}{\mathcal{Z}[E, 0]}$, the connected $\displaystyle \sum_{i=1}^{B}N_{i}$-point function $G_{|a_{1}^{1}\ldots a_{N_{1}}^{1}|\ldots|a_{1}^{B}\ldots a_{N_{B}}^{B}|}$ is defined as
\begin{align}
\log\frac{ \mathcal{Z}[E, J]}{\mathcal{Z}[E, 0]}
:=\sum_{B=1}^\infty \sum_{1\leq N_1 \leq \dots \leq
  N_B}^\infty
\sum_{p_1^1,\dots,p^B_{N_B} =0}^{\mathcal{N}} \!\!\!\!
N^{2-B}
&\frac{G_{|p_1^1\dots p_{N_1}^1|\dots|p_1^B\dots p^B_{N_B}|}
}{S_{(N_1,\dots ,N_B)}}
\prod_{\beta=1}^B \frac{\mathbb{J}_{p_1^\beta\dots
    p^\beta_{N_\beta}}}{N_\beta},
\label{logZ}
\end{align}
where $N_{i}$ is the identical valence number for $i=1,\ldots,B$, $\displaystyle\mathbb{J}_{p_1\dots p_{N_{i}}}:=\prod_{j=1}^{N_{i}} J_{p_jp_{j+1}}$ with $N_{i}+1\equiv 1$, $(N_1,\dots,N_B)=(\underbrace{N'_1,\dots,N'_1}_{\nu_1},\dots,
\underbrace{N'_s,\dots,N'_s}_{\nu_s})$, and $\displaystyle S_{(N_1,\dots ,N_B)}=\prod_{\beta=1}^{s} \nu_{\beta}!$. The $\displaystyle \sum_{i=1}^{B}N_{i}$-point function denoted by $G_{|a_{1}^{1}\ldots a_{N_{1}}^{1}|\ldots|a_{1}^{B}\ldots a_{N_{B}}^{B}|}$ is given by the sum over all Feynman diagrams (ribbon graphs) on Riemann surfaces with $B$-boundaries, and each $|a^{i}_{1}\cdots a^{i}_{N_{i}}|$ corresponds to the Feynman diagrams having $N_{i}$-external ribbons from the $i$-th boundary\cite{Grosse:2012uv}. (See Figure \ref{qwe}.)
\begin{figure}[h!]
\begin{center}
\includegraphics[width=120mm]{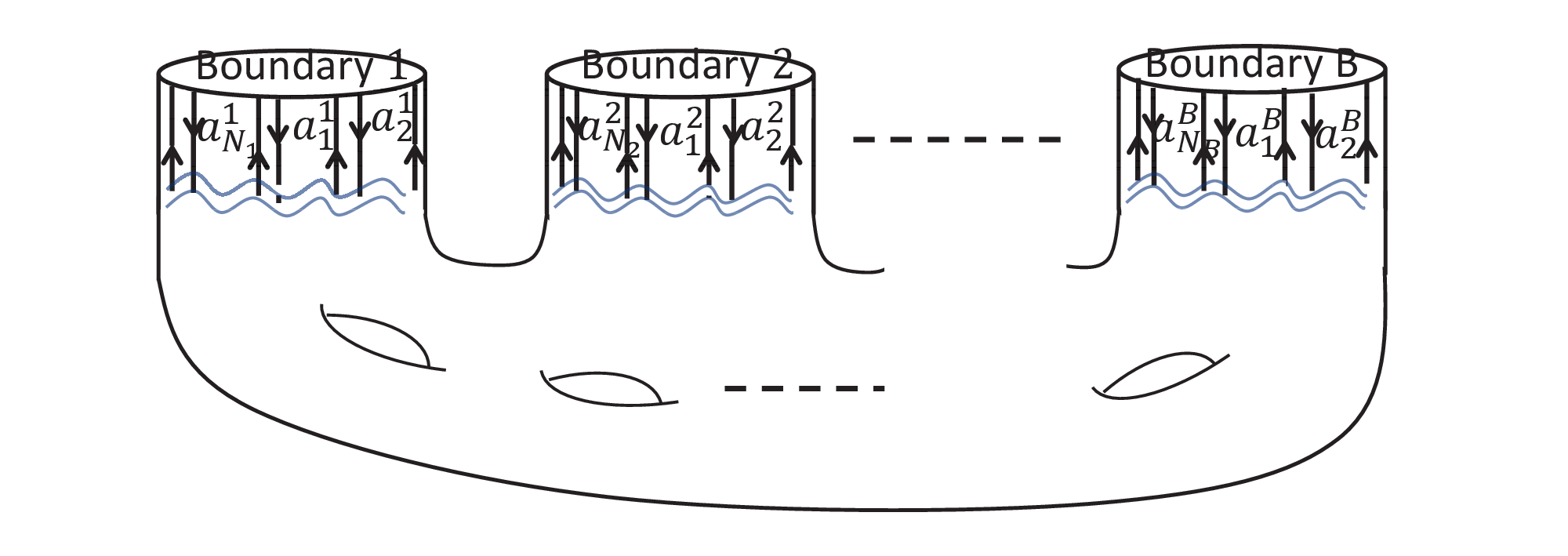}
\caption{The relationship between external ribbons of Feynman diagrams and boundaries as expressed in $G_{|a_{1}^{1}\ldots a_{N_{1}}^{1}|\ldots|a_{1}^{B}\ldots a_{N_{B}}^{B}|}$ .}\label{qwe}
\end{center}
\end{figure}
Some connected $\sum_{i=1}^B N_i$ points function
$\langle\Phi_{a_{1}^{1}a_{2}^{1}}\cdots\Phi_{a_{N_{1}}^{1}a_{{1}}^{1}}\Phi_{a_{1}^{2}a_{2}^{2}}\cdots\Phi_{a_{N_{2}}^{2}a_{{1}}^{2}}\cdots\Phi_{a_{1}^{B}a_{2}^{B}}\cdots\Phi_{a_{N_{B}}^{B}a_{{1}}^{B}}\rangle_{c}$ might include contributions from several types of surfaces classified by their boundaries. 
For example, let us consider $\langle\Phi_{aa}\Phi_{aa}\rangle_{c}$. From (\ref{logZ}), 
$\displaystyle\langle\Phi_{aa}\Phi_{aa}\rangle_{c}=\frac{1}{N}G_{|aa|}+\frac{1}{N^{2}}G_{|a|a|}$. This means that $\langle\Phi_{aa}\Phi_{aa}\rangle_{c}$ includes contributions from two types of surfaces which are surfaces with one boundary and ones with two boundaries. \bigskip

We prepare a connected oriented surface with $B$ boundaries for drawing each Feynman diagram to calculate $G_{|a_{1}^{1}\ldots a_{N_{1}}^{1}|\ldots|a_{1}^{B}\ldots a_{N_{B}}^{B}|}$. We draw a Feynman diagram with external ribbons with $(a_{1}^{i}a_{2}^{i}),\cdots,(a_{N_{i}}^{i}a_{1}^{i})$ subscripted to each boundary $i$. For any connected segments in a Feynman diagram, both ends are on the same boundary. $G_{|a_{1}^{1}\ldots a_{N_{1}}^{1}|\ldots|a_{1}^{B}\ldots a_{N_{B}}^{B}|}$ is given by the sum over all such Feynman diagrams.


Let $F[E, J]:=\log\mathcal{Z}[E, J]$ be the free energy. 
Then the terms up to the fourth order for $J$ are written out explicitly
as follows:
\begin{align}
F[E, J]=&F[E,0]+\sum_{B=1}^{\infty}\sum_{N_{1},\cdots,N_{B}=1}^{\infty}\sum_{p_{1}^{1},\cdots,p_{N_{B}}^{B}=1}^{N}N^{2-B}\frac{G_{|p_{1}^{1}\cdots p_{N_{1}}^{1}|\cdots|p_{1}^{B}\cdots p_{N_{B}}^{B}|}}{B!}\prod_{\beta=1}^{B}\frac{\mathbb{J}_{p_{1}^{\beta}\cdots p_{N_{\beta}}^{\beta}}}{N_{\beta}}\notag\\
=&F[E,0]+\frac{N}{2}\sum_{m,n=1}^{N}G_{|mn|}J_{mn}J_{nm}+\frac{1}{2}\sum_{m,n=1}^{N}G_{|m|n|}J_{mm}J_{nn}\notag\\
&+\frac{1}{24N^{2}}\sum_{m,n,k,r=1}^{N}G_{|m|n|k|r|}J_{mm}J_{nn}J_{kk}J_{rr}\notag\\
&+\frac{1}{4N}\sum_{m,n,k,r=1}^{N}G_{|m|n|kr|}J_{mm}J_{nn}J_{kr}J_{rk}+\frac{1}{8}\sum_{m,n,k,r=1}^{N}G_{|mn|kr|}J_{mn}J_{nm}J_{kr}J_{rk}\notag\\
&+\frac{1}{3}\sum_{m,n,k,r=1}^{N}G_{|m|nkr|}J_{mm}J_{nk}J_{kr}J_{rn}+\frac{N}{4}\sum_{m,n,k,r=1}^{N}G_{|mnkr|}J_{mn}J_{nk}J_{kr}J_{rm}+\mathcal{O}(J^{6}).
\end{align}
We have taken into account that a correlation function 
$G_{|a_{1}^{1}\ldots a_{N_{1}}^{1}|\ldots|a_{1}^{B}\ldots a_{N_{B}}^{B}|}$
in the $\Phi^{4}$ matrix model
is identically zero if $\sum_{i=1}^{B}N_{i}$ is odd.

The Feynman rules are given as follows.
 the propagator for the free Lagrangian is given by
\begin{align}
\langle \Phi_{ab}\Phi_{cd} \rangle_f 
=  \frac{1}{N} \frac{\delta_{ad} \delta_{bc}}{E_a + E_b}
=:  \lower1.3ex\hbox{\propagator} ,
\end{align}
the interaction is given as 
\begin{align}
 \lower1.9ex\hbox{\vertexs}:= -\frac{\eta}{4} \times 
 \mbox{symmetry factor},
\end{align}
and each loop corresponds to a sum $\sum_n$.


\subsection{Formula for connected Green's function }

We will rewrite the Schwinger-Dyson equation (\ref{SDeq_FreeEnergy}) using a connected Green's function.
First, the first-order derivative terms are calculated.
Using 
\begin{align}
\frac{\partial}{\partial E_{i}}F[E, J]=&
-\frac{1}{N\mathcal{Z}[E, J]}
\sum_{k=1}^{N}\frac{\partial^{2}}{\partial J_{ik}\partial J_{ki}} \mathcal{Z}[E, J] ,
\end{align}
we find
\begin{align}
\frac{\partial}{\partial E_{i}}F[E,0]=
&\left.-\frac{1}{N\mathcal{Z}[E,0]}\left\{\sum_{k=1}^{N}\left(\frac{\partial^{2}}{\partial J_{ik}\partial J_{ki}}F[E,J]\right)e^{F[E,J]}\right\}\right|_{J=0}\notag\\
=&-\sum_{k=1}^{N}G_{|ik|}-\frac{1}{N}G_{|i|i|}. \label{z}
\end{align}
Using (\ref{z}), we obtain
\begin{align}
\frac{\eta}{N}\sum_{i=1}^{N}\left(\frac{\partial}{\partial E_{i}}F[E,0]\right)\left(\frac{\partial}{\partial E_{i}}F[E,0]\right)
=&\frac{\eta}{N}\sum_{i,k,l=1}^{N}G_{|ik|}G_{|il|}+\frac{2\eta}{N^{2}}\sum_{i,k=1}^{N}G_{|ik|}G_{|i|i|}+\frac{\eta}{N^{3}}\sum_{i=1}^{N}G_{|i|i|}^{2} ,
\end{align}
and
%
\begin{align}
&\frac{\eta}{N}\sum_{i,j=1,i\neq j}^{N}\frac{1}{E_{i}-E_{j}}\left(\frac{\partial}{\partial E_{i}}F[E,0] 
-\frac{\partial}{\partial E_{i}}F[E,0] 
\right)\notag \\
&=-\frac{\eta}{N}\sum_{i,j,k=1,i\neq j}^{N}\frac{1}{E_{i}-E_{j}}(G_{|ik|}
- G_{|jk|})
-\frac{\eta}{N^{2}}\sum_{i,j=1,i\neq j}^{N}\frac{1}{E_{i}-E_{j}}(G_{|i|i|} - G_{|j|j|}) .
\end{align}

Next let us transform the Laplacian term 
$\displaystyle \frac{\eta}{N}\sum_{i=1}^{N}\left(\frac{\partial^{2}}{\partial E_{i}^{2}}F[E,J]\right)$
to express it in terms of a connected Green's function. Note that
\begin{align}
&\frac{\eta}{N}\sum_{i=1}^{N}\left(\frac{\partial^{2}}{\partial E_{i}^{2}}F [E,J] \right) \\
&=\frac{\eta}{N}\sum_{i=1}^{N}\left\{\frac{\frac{1}{N^{2}}
{\sum_{k,l=1}^{N}\frac{\partial^{2}}{\partial J_{ik}\partial J_{ki}}\frac{\partial^{2}}{\partial J_{il}\partial J_{li}}\mathcal{Z}[E,J]}}{\mathcal{Z} [E,J] }-\left(\frac{-\frac{1}{N}{\sum_{k=1}^{N}\frac{\partial^{2}}{\partial J_{ik}\partial J_{ki}}\mathcal{Z}[E,J]}}{\mathcal{Z} [E,J] }\right)^{2}\right\} , \notag
\end{align}
and one point functions vanish.
%
The second term with $J=0$  is given as
\begin{align}
&\left. -\frac{\eta}{N}\sum_{i=1}^{N}\left\{\left(\frac{-\frac{1}{N}{\sum_{k=1}^{N}\frac{\partial^{2}}{\partial J_{ik}\partial J_{ki}}\mathcal{Z}[E,J]}}{\mathcal{Z} [E,J] }\right)^{2}\right\} \right|_{J=0}
=\left.-\frac{\eta}{N^{3}(\mathcal{Z}[E, 0])^{2}}\sum_{i=1}^{N}\left\{\sum_{k=1}^{N}\left(\frac{\partial^{2}F[E,J]}{\partial J_{ik}\partial J_{ki}}\right)e^{F[E,J]}\right\}^{2}\right|_{J=0}\notag\\
&=-\frac{\eta}{N}\sum_{i,l,t=1}^{N}G_{|il|}G_{|it|}-\frac{2\eta}{N^{2}}\sum_{i,l=1}^{N}G_{|il|}G_{|i|i|}-\frac{\eta}{N^{3}}\sum_{i=1}^{N}G_{|i|i|}^{2} .
\end{align}
The first term with $J=0$  is complex:
\begin{align}
&\left. {\sum_{k,l=1}^{N}\frac{\partial^{2}}{\partial J_{ik}\partial J_{ki}}\frac{\partial^{2}}{\partial J_{il}\partial J_{li}}\mathcal{Z}[E,J]} \right|_{J=0}\notag\\
=&\sum_{k,l=1}^{N}\Biggl\{ {\left(\frac{\partial^{4}}{\partial J_{il}\partial J_{li}\partial J_{ik}\partial J_{ki}}F[E,J]\right)} e^{F[E,J]}
+{\left(\frac{\partial^{2}}{\partial J_{ik}\partial J_{ki}}F[E,J]\right)\left(\frac{\partial^{2}}{\partial J_{il}\partial J_{li}}F[E,J]\right)} e^{F[E,J]}\notag\\
&+{\left(\frac{\partial^{2}}{\partial J_{li}\partial J_{ki}}F[E,J]\right)\left(\frac{\partial^{2}}{\partial J_{il}\partial J_{ik}}F[E,J]\right)} e^{F[E,J]}\notag\\
&\left.+{\left(\frac{\partial^{2}}{\partial J_{il}\partial J_{ki}}F[E,J]\right)\left(\frac{\partial^{2}}{\partial J_{li}\partial J_{ik}}F[E,J]\right)} e^{F[E,J]}\Biggl\}\right|_{J=0} . \label{ZJJJJ}
\end{align}

%
For example, the first term of (\ref{ZJJJJ}) is obtained as
\begin{align}
&\left.\frac{\eta}{N^{3}}\sum_{i,l,m=1}^{N}\frac{\partial^{4}F[E,J]}{\partial J_{im}\partial J_{mi}\partial J_{il}\partial J_{li}} \right|_{J=0}\notag\\
=&\frac{\eta}{N^{5}}\sum_{i=1}^{N}G_{|i|i|i|i|}+\frac{2\eta}{N^{4}}\sum_{i,l=1}^{N}G_{|i|i|il|}+\frac{4\eta}{N^{4}}\sum_{i=1}^{N}G_{|i|i|ii|}+\frac{\eta}{N^{3}}\sum_{i,m,l=1}^{N}G_{|im|il|}+\frac{\eta}{N^{3}}\sum_{i=1}^{N}G_{|ii|ii|}\notag\\
&+\frac{\eta}{N^{3}}\sum_{l,i=1}^{N}G_{|il|il|}+\frac{4\eta}{N^{3}}\sum_{i=1}^{N}G_{|i|iii|}+\frac{4\eta}{N^{3}}\sum_{i,l=1}^{N}G_{|i|iil|}+\frac{2\eta}{N^{2}}\sum_{i=1}^{N}G_{|iiii|}+\frac{2\eta}{N^{2}}\sum_{i,l=1}^{N}G_{|ilii|}\notag\\
&+\frac{\eta}{N^{2}}\sum_{l,i=1}^{N}G_{|ilil|}+\frac{\eta}{N^{2}}\sum_{i,m,l=1}^{N}G_{|limi|} .\label{x}
\end{align}

After similar calculations, $\displaystyle\frac{\eta}{N}\sum_{i=1}^{N}\left(\frac{\partial^{2}}{\partial E_{i}^{2}}F[E,0]\right)$ is obtained as follows:
\begin{align}
%
&
\frac{\eta}{N^{5}}\sum_{i=1}^{N}G_{|i|i|i|i|}+\frac{2\eta}{N^{4}}\sum_{i,l=1}^{N}G_{|i|i|il|}+\frac{4\eta}{N^{4}}\sum_{i=1}^{N}G_{|i|i|ii|}+\frac{\eta}{N^{3}}\sum_{i,m,l=1}^{N}G_{|im|il|}\notag\\
&+\frac{\eta}{N^{3}}\sum_{i=1}^{N}G_{|ii|ii|}+\frac{\eta}{N^{3}}\sum_{i,l=1}^{N}G_{|il|il|}+\frac{4\eta}{N^{3}}\sum_{i=1}^{N}G_{|i|iii|}+\frac{4\eta}{N^{3}}\sum_{i,l=1}^{N}G_{|i|iil|}\notag\\
&+\frac{2\eta}{N^{2}}\sum_{i=1}^{N}G_{|iiii|}+\frac{2\eta}{N^{2}}\sum_{i,l=1}^{N}G_{|ilii|}+\frac{\eta}{N^{2}}\sum_{i,l=1}^{N}G_{|ilil|}+\frac{\eta}{N^{2}}\sum_{i,m,l=1}^{N}G_{|limi|}\notag\\
&+\frac{\eta}{N}\sum_{i=1}^{N}G_{|ii|}^{2}+\frac{4\eta}{N^{2}}\sum_{i=1}^{N}G_{|ii|}G_{|i|i|}+\frac{2\eta}{N^{3}}\sum_{i=1}^{N}G_{|i|i|}^{2}+\frac{\eta}{N}\sum_{i,m=1}^{N}G_{|im|}^{2} .
\label{FEE}
\end{align}

Summarizing the results from (\ref{z}) to (\ref{FEE}),
the equivalent equation with $\mathcal{L}_{SD}\mathcal{Z}[E,0]=\mathcal{L}_{SD}e^{F[E,0]}=0$
in the form of connected Green function is obtained.
\begin{proposition} \label{SD_eq_Green}
The connected $4$-point functions and the connected $2$-point
functions defined in (\ref{logZ})
satisfy the following relation.
\begin{align}
&\frac{\eta}{N^{5}}\sum_{i=1}^{N}G_{|i|i|i|i|}+\frac{2\eta}{N^{4}}\sum_{i,l=1}^{N}G_{|i|i|il|}+\frac{4\eta}{N^{4}}\sum_{i=1}^{N}G_{|i|i|ii|}+\frac{\eta}{N^{3}}\sum_{i,m,l=1}^{N}G_{|im|il|}+\frac{\eta}{N^{3}}\sum_{i=1}^{N}G_{|ii|ii|}\notag\\
&+\frac{\eta}{N^{3}}\sum_{i,l=1}^{N}G_{|il|il|}+\frac{4\eta}{N^{3}}\sum_{i=1}^{N}G_{|i|iii|}+\frac{4\eta}{N^{3}}\sum_{i,l=1}^{N}G_{|i|iil|}+\frac{2\eta}{N^{2}}\sum_{i=1}^{N}G_{|iiii|}+\frac{2\eta}{N^{2}}\sum_{i,l=1}^{N}G_{|ilii|}\notag\\
&+\frac{\eta}{N^{2}}\sum_{i,l=1}^{N}G_{|ilil|}+\frac{\eta}{N^{2}}\sum_{i,m,l=1}^{N}G_{|limi|}\notag\\
&+\frac{\eta}{N}\sum_{i,l,m=1}^{N}G_{|il|}G_{|im|}+\frac{2\eta}{N^{2}}\sum_{i,l=1}^{N}G_{|il|}G_{|i|i|}+\frac{3\eta}{N^{3}}\sum_{i=1}^{N}G_{|i|i|}^{2}+\frac{\eta}{N}\sum_{i=1}^{N}G_{|ii|}^{2}+\frac{4\eta}{N^{2}}\sum_{i=1}^{N}G_{|ii|}G_{|i|i|}\notag\\
&+\frac{\eta}{N}\sum_{i,k=1}^{N}G_{|ik|}^{2}+\frac{\eta}{N}\sum_{i,j,k=1,i\neq j}^{N}\frac{1}{E_{i}-E_{j}}\left(G_{|jk|}-G_{|ik|}\right)+\frac{\eta}{N^{2}}\sum_{i,j=1,i\neq j}^{N}\frac{1}{E_{i}-E_{j}}(G_{|j|j|}-G_{|i|i|})\notag\\
&+2\sum_{i,k=1}^{N}E_{i}G_{|ik|}+\frac{2}{N}\sum_{i=1}^{N}E_{i}G_{|i|i|}-N^{2}=0 . \label{u}
\end{align}

\end{proposition}

\subsection{Perturbative Check for Schwinger-Dyson Equation}
Let us perturbatively check Proposition \ref{SD_eq_Green}
up to the first order of $\eta$ in this subsection.\\

First we calculate $G_{|ik|}$.
\begin{align}
G_{|ik|}=&\frac{1}{N}\left.\frac{\partial^{2}}{\partial J_{ik}\partial J_{ki}}\log\frac{\mathcal{Z}[E,J]}{\mathcal{Z}[E,0]}\right|_{J=0}
=\left.\frac{1}{N\mathcal{Z}[E,0]}\frac{\partial^{2}\mathcal{Z}[E,J]}{\partial J_{ik}\partial J_{ki}}\right|_{J=0}\notag\\
=&\frac{1}{N\mathcal{Z}[E,0]}\frac{\partial^{2}}{\partial J_{ik}\partial J_{ki}}\Biggl\{\sum_{t=0}^{\infty}\sum_{l=0}^{\infty}\frac{1}{t!}\frac{1}{l!}\int\mathcal{D}\Phi\left(-\frac{N\eta}{4}\right)^{t}\left(\sum_{n_{1},n_{2},n_{3},n_{4}=1}^{N}\Phi_{n_{1}n_{2}}\Phi_{n_{2}n_{3}}\Phi_{n_{3}n_{4}}\Phi_{n_{4}n_{1}}\right)^{t}\notag\\
&\left.\times N^{l}\left(\sum_{m_{1},m_{2}=1}^{N}J_{m_{1}m_{2}}\Phi_{m_{2}m_{1}}\right)^{l}\exp\left(-N\mathrm{tr}\left(E\Phi^{2}\right)\right)\Biggl\}\right|_{J=0}\notag\\
=&\frac{1}{E_{k}+E_{i}}-\frac{\eta}{N}\sum_{n_{3}=1}^{N}\frac{1}{(E_{k}+E_{i})^{2}(E_{n_{3}}+E_{i})}-\frac{\eta}{N}\sum_{n_{3}=1}^{N}\frac{1}{(E_{k}+E_{i})^{2}(E_{n_{3}}+E_{k})}+\mathcal{O}(\eta^{2})\label{v}
\end{align}

Second we calculate $G_{|i|k|}$.
\begin{align}
G_{|i|k|}=&\left.\frac{\partial^{2}}{\partial J_{ii}\partial J_{kk}}\log\frac{\mathcal{Z}[E,J]}{\mathcal{Z}[E,0]}\right|_{J=0}
=\left.\frac{1}{\mathcal{Z}[E,0]}\frac{\partial^{2}\mathcal{Z}[E,J]}{\partial J_{ii}\partial J_{kk}}\right|_{J=0}\notag\\
=&\frac{1}{\mathcal{Z}[E,0]}\frac{\partial^{2}}{\partial J_{ii}\partial J_{kk}}\Biggl\{\sum_{t=0}^{\infty}\sum_{l=0}^{\infty}\frac{1}{t!}\frac{1}{l!}\int\mathcal{D}\Phi\left(-\frac{N\eta}{4}\right)^{t}\left(\sum_{n_{1},n_{2},n_{3},n_{4}=1}^{N}\Phi_{n_{1}n_{2}}\Phi_{n_{2}n_{3}}\Phi_{n_{3}n_{4}}\Phi_{n_{4}n_{1}}\right)^{t}\notag\\
&\left.\times N^{l}\left(\sum_{m_{1},m_{2}=1}^{N}J_{m_{1}m_{2}}\Phi_{m_{2}m_{1}}\right)^{l}\exp\left(-N\mathrm{tr}\left(E\Phi^{2}\right)\right)\Biggl\}\right|_{J=0}\notag\\
=&-\frac{\eta}{4E_{i}E_{k}(E_{k}+E_{i})}+\mathcal{O}(\eta^{2})\label{w}
\end{align}

Note that from the Feynmann rule given in the end of Subsection
\ref{subset5_1}, the following terms in Proposition \ref{SD_eq_Green}
is $\mathcal{O}(\eta^{2})$:
\begin{align*}
\mathcal{O}(\eta^{2})=&\frac{\eta}{N^{5}}\sum_{i=1}^{N}G_{|i|i|i|i|}+\frac{2\eta}{N^{4}}\sum_{i,l=1}^{N}G_{|i|i|il|}+\frac{4\eta}{N^{4}}\sum_{i=1}^{N}G_{|i|i|ii|}+\frac{\eta}{N^{3}}\sum_{i,m,l=1}^{N}G_{|im|il|}+\frac{\eta}{N^{3}}\sum_{i=1}^{N}G_{|ii|ii|}\notag\\
&+\frac{\eta}{N^{3}}\sum_{i,l=1}^{N}G_{|il|il|}+\frac{4\eta}{N^{3}}\sum_{i=1}^{N}G_{|i|iii|}+\frac{4\eta}{N^{3}}\sum_{i,l=1}^{N}G_{|i|iil|}+\frac{2\eta}{N^{2}}\sum_{i=1}^{N}G_{|iiii|}+\frac{2\eta}{N^{2}}\sum_{i,l=1}^{N}G_{|ilii|}\notag\\
&+\frac{\eta}{N^{2}}\sum_{i,l=1}^{N}G_{|ilil|}+\frac{\eta}{N^{2}}\sum_{i,m,l=1}^{N}G_{|limi|}+\frac{\eta}{N^{2}}\sum_{i,j=1,i\neq j}^{N}\frac{1}{E_{i}-E_{j}}(G_{|j|j|}-G_{|i|i|})\notag\\
&+\frac{4\eta}{N^{2}}\sum_{i=1}^{N}G_{|ii|}G_{|i|i|}+\frac{2\eta}{N^{2}}\sum_{i,l=1}^{N}G_{|il|}G_{|i|i|}+\frac{3\eta}{N^{3}}\sum_{i=1}^{N}G_{|i|i|}^{2} .
\end{align*}

From (\ref{w}),(\ref{v}), we can calculate perturbatively (\ref{u}) 
as follows
\begin{align*}
&L.H.S. of  (\ref{u}) \notag \\
=&\Biggl\{-N^{2}+\frac{2}{N}\sum_{i=1}^{N}E_{i}G_{|i|i|}+2\sum_{i,k=1}^{N}E_{i}G_{|ik|}+\frac{\eta}{N}\sum_{i,j,k=1,i\neq j}^{N}\frac{1}{E_{i}-E_{j}}\left(G_{|jk|}-G_{|ik|}\right)\notag\\
&+\frac{\eta}{N}\sum_{i=1}^{N}G_{|ii|}^{2}
+\frac{\eta}{N}\sum_{i,l,m=1}^{N}G_{|il|}G_{|im|}+\frac{\eta}{N}\sum_{i,k=1}^{N}G_{|ik|}^{2}\Biggl\} +\mathcal{O}(\eta^{2}) \notag\\
=&-N^{2}-\frac{\eta}{4N}\sum_{i=1}^{N}\frac{1}{E_{i}^{2}}+2\sum_{i,k=1}^{N}E_{i}\frac{1}{E_{k}+E_{i}}-\frac{2\eta}{N}\sum_{i,k,n_{3}=1}^{N}\frac{E_{i}}{(E_{k}+E_{i})^{2}(E_{n_{3}}+E_{i})}\notag\\
&-\frac{2\eta}{N}\sum_{i,k,n_{3}=1}^{N}\frac{E_{i}}{(E_{k}+E_{i})^{2}(E_{n_{3}}+E_{k})}+\frac{\eta}{N}\sum_{i,k,j=1,i\neq j}^{N}\frac{1}{(E_{j}+E_{k})(E_{i}+E_{k})}+\frac{\eta}{4N}\sum_{i=1}^{N}\frac{1}{E_{i}^{2}}\notag\\
&+\frac{\eta}{N}\sum_{i,l,m=1}^{N}\frac{1}{(E_{i}+E_{l})(E_{i}+E_{m})}+\frac{\eta}{N}\sum_{i,k=1}^{N}\frac{1}{(E_{i}+E_{k})^{2}}+\mathcal{O}(\eta^{2})\notag\\
=& 0+\mathcal{O}(\eta^{2}) .
\end{align*}
Thus, it was also confirmed that there is no inconsistency in the perturbation calculation for $\eta$
up to order $\eta^2$.

\section{Summary}
It was recently discovered that the partition function of the $\Phi^4$ matrix model with a Kontsevich-type kinetic term satisfies the Schr\"odinger equation of the $N$-body harmonic oscillator, and that eigenstates of the Virasoro operators can be derived from this partition function. 
In this paper, we build upon these findings and obtain an explicit formula for such eigenstates in terms of the free energy, as demonstrated in Section \ref{sect3}. Furthermore, since the free energy serves as the generating function for connected multi-point correlation functions, the differential equation for the harmonic oscillator can also be reformulated in terms of these connected correlators. The corresponding equations for the connected two- and four-point functions are derived in Section 5. These results are further confirmed perturbatively up to first order in the coupling constant of the interaction.

The process of obtaining the Schr\"odinger equation for the $N$-body harmonic oscillator is constructed from a set of Schwinger-Dyson equations.
The contribution of additional Schwinger-Dyson equations is often discussed using loop equations in matrix models with $U(N)$ symmetry.
Although this model lacks $U(N)$ symmetry due to the presence of a kinetic term, it retains $U(1)^N$ symmetry, enabling us to derive equations similar to loop equations, as described in \cite{Eynard:2015aea}. This is done in Section \ref{sect4}.\\


%
\section*{Acknowledgements}
\noindent 
A.S.\ was supported by JSPS KAKENHI Grant Number 21K03258.
R.W. was supported by the Deutsche Forschungsgemeinschaft (DFG, German
Research Foundation) --  Project-ID 427320536 -- SFB 1442, as well as
under Germany's Excellence Strategy EXC 2044 -- 390685587,
Mathematics M\"unster:  Dynamics -- Geometry -- Structure.
We would like to thank the Erwin Schr\"odinger International Institute for Mathematics and Physics (ESI) for supporting the realization of this collaborative research.
\\


\noindent
{\bf Data availability} \   Data sharing is not applicable to this article as no new data were 
created or analyzed in this study.

\section*{Declarations}

{\bf Conflicts of interest} \ On behalf of all authors, the corresponding author states that there is no conflict of
interest.

\appendix

\section{Notations}
For the reader's convenience, we provide a list of notations in this appendix. 
The two coordinate systems, 
$E_i$ and $y_i$, are related by $\displaystyle y_{i}=\sqrt{\frac{N}{\eta}}E_{i}$.
\begin{itemize}
\item$\displaystyle\mathcal{L}_{SD}:=\frac{\eta}{N}\sum_{i=1}^{N}\left(\frac{\partial}{\partial E_{i}}\right)^{2}\!\!\!+\frac{\eta}{N}\!\!\!\sum_{i,j=1,i\neq j}^{N}\frac{1}{E_{i}-E_{j}}\left(\frac{\partial}{\partial E_{i}}-\frac{\partial}{\partial E_{j}}\right)-2\sum_{k=1}^{N}E_{k}\frac{\partial}{\partial E_{k}}-N^{2}$
\item$\mathcal{L}_{SD}=\sum_{i=1}^{N}\left(\frac{\partial}{\partial y_{i}}\right)^{2}+\sum_{i,j=1,i\neq j}^{N}\frac{1}{y_{i}-y_{j}}\left(\frac{\partial}{\partial y_{i}}-\frac{\partial}{\partial y_{j}}\right)-2\sum_{i=1}^{N}y_{i}\frac{\partial}{\partial y_{i}}-N^{2}$
\vspace{2mm}
\item$F[E,0]=\log\mathcal{Z}[E,0]$\hspace{2mm}: Free energy
\vspace{2mm}
\item$\hat{\mathcal{L}}_{SD}=e^{-F[E,0]}\mathcal{L}_{SD}e^{F[E,0]}$
\vspace{2mm}
\item$a_{i}=\frac{1}{\sqrt{2}}\left(y_{i}+\frac{\partial}{\partial y_{i}}\right)=\frac{1}{\sqrt{2}}e^{-\frac{1}{2}\sum_{k=1}^{N}y_{k}^{2}}\frac{\partial}{\partial y_{i}}e^{\frac{1}{2}\sum_{k=1}^{N}y_{k}^{2}}$
\vspace{2mm}
\item$a_{i}^{\dagger}=\frac{1}{\sqrt{2}}\left(y_{i}-\frac{\partial}{\partial y_{i}}\right)=-\frac{1}{\sqrt{2}}e^{\frac{1}{2}\sum_{k=1}^{N}y_{k}^{2}}\frac{\partial}{\partial y_{i}}e^{-\frac{1}{2}\sum_{k=1}^{N}y_{k}^{2}}$
\vspace{2mm}
\item$L_{-m}:=\sum_{i=1}^{N}\left(\alpha(a_{i}^{\dagger})^{m+1}a_{i}+(1-\alpha)a_{i}(a_{i}^{\dagger})^{m+1}\right)$\\
$=\sum_{i=1}^{N}\biggl\{a_{i}(a_{i}^{\dagger})^{m+1}-\alpha(m+1)(a_{i}^{\dagger})^{m}\biggl\}$
\vspace{2mm}
\item$\Delta(E)=\prod_{1\leq i<j\leq N}(E_{j}-E_{i})$
\vspace{2mm}
\item$\Delta(y)=\prod_{1\leq i<j\leq N}(y_{j}-y_{i})=\exp\left(\log\prod_{1\leq i<j\leq N}(y_{j}-y_{i})\right)=\exp\left(\sum_{1\leq i<j\leq N}\log(y_{j}-y_{i})\right)$
\vspace{2mm}
\item$g:=e^{\frac{N}{2\eta}\sum_{i=1}^{N}E_{i}^{2}}\Delta^{-1}(E)=\left(\frac{N}{\eta}\right)^{\frac{N(N-1)}{4}}e^{\frac{1}{2}\sum_{i=1}^{N}y_{i}^{2}}\Delta^{-1}(y)$
\vspace{2mm}
\item$\widetilde{L}_{-m}:=gL_{-m}g^{-1}$
\vspace{2mm}
\item$\hat{L}_{-m}=e^{-F}\widetilde{L}_{-m}e^{F}$
\vspace{2mm}
\item$e^{-\frac{N}{2\eta}\sum_{i=1}^{N}E_{i}^{2}}\Delta(E)\mathcal{L}_{SD}\Delta^{-1}(E)e^{\frac{N}{2\eta}\sum_{i=1}^{N}E_{i}^{2}}=\frac{\eta}{N}\sum_{i=1}^{N}\left(\frac{\partial}{\partial E_{i}}\right)^{2}-\frac{N}{\eta}\sum_{i=1}^{N}(E_{i})^{2}=-\mathcal{H}_{HO}$
\end{itemize}

\end{document}